\documentclass[11pt]{article}
\usepackage{lmodern}
\usepackage{xspace}
\usepackage{amssymb}
\usepackage{amsfonts}
\usepackage{amsmath}
\usepackage{amsthm}
\usepackage{latexsym}
\usepackage[center]{subfigure}
\usepackage{epsfig}

\usepackage{filecontents}
\usepackage[ruled]{algorithm2e}

\newtheorem{theorem}{Theorem}[section]
\newtheorem{lemma}[theorem]{Lemma}

\newtheorem{proposition}[theorem]{Proposition}
\newtheorem{claim}[theorem]{Claim}

\newcommand{\cald}{{\cal D}}

\newcommand{\caln}{{\cal N}}
\newcommand{\cals}{{\cal S}}
\newcommand{\R}{\mathbb{R}}
\newcommand{\N}{\mathbb{R}}
\newcommand{\C}{\mathbb{C}}
\newcommand{\Z}{\mathbb{Z}}

\newcommand{\poly}{\mathrm{poly}}

\newcommand{\eps}{\ensuremath{\epsilon}\xspace}

\newcommand{\eqdef}{\stackrel{\rm def}{=}}

\newcommand{\totalvardist}[2]{{\operatorname{d_{\rm TV}}\!\left({#1, #2}\right)}}

\newcommand{\norm}[1]{\lVert#1{\rVert}}
\newcommand{\normone}[1]{{\norm{#1}}_1}
\newcommand{\normtwo}[1]{{\norm{#1}}_2}
\newcommand{\norminf}[1]{{\norm{#1}}_\infty}
\newcommand{\abs}[1]{\left\lvert #1 \right\rvert}
\newcommand{\setOfSuchThat}[2]{ \left\{\; #1 \;\colon\; #2\; \right\} } 			

\newcommand{\expect}[1]{\mathbb{E}\!\left[#1\right]}
\DeclareMathOperator*{\expectOp}{\mathbb{E}}
\newcommand{\expectOver}[2]{\expectOp_{#1}\!\left[#2\right]}

\newcommand{\var}{\operatorname{Var}}

\newcommand{\bigO}[1]{{O\!\left({#1}\right)}}
\newcommand{\bigTheta}[1]{{\Theta\!\left({#1}\right)}}
\newcommand{\bigOmega}[1]{{\Omega\!\left({#1}\right)}}
\newcommand{\tildeO}[1]{\tilde{O}\!\left({#1}\right)}

\usepackage[shortlabels]{enumitem}

\usepackage{fullpage}
\usepackage{bm}
\renewcommand{\vec}[1]{\bm{#1}}
\newcommand{\bone}{\vec{1}}
\renewcommand{\tilde}{\widetilde}

\newcommand{\wh}[1]{{\widehat{#1}}}
\newcommand{\dtv}{d_{\mathrm TV}}

\begin{document}

\title{Playing Anonymous Games using Simple Strategies}
\author{Yu Cheng\thanks{Supported in part by Shang-Hua Teng's Simons Investigator Award.}
  \qquad
  Ilias Diakonikolas\thanks{Supported by a USC startup fund.}
  \qquad Alistair Stewart\footnote{Authors' emails: \texttt{\{yu.cheng.1, diakonik, alistais\}@usc.edu}}
  \\ University of Southern California}
\date{}

\begin{titlepage}
\clearpage\maketitle
\thispagestyle{empty}

\begin{abstract}
We investigate the complexity of computing approximate Nash equilibria in anonymous games.
Our main algorithmic result is the following:
For any $n$-player anonymous game with a bounded number of strategies and any constant $\delta>0$,
an $O(1/n^{1-\delta})$-approximate Nash equilibrium can be computed in polynomial time.
Complementing this positive result, we show that if there exists any constant $\delta>0$
such that an $\bigO{1/n^{1+\delta}}$-approximate equilibrium can be computed in polynomial time,
then
there is a fully polynomial-time approximation scheme for this problem.

We also present a faster algorithm that, for any $n$-player $k$-strategy anonymous game,
runs in time $\tildeO{(n+k) k n^k}$ and computes an $\tildeO{n^{-1/3} k^{11/3}}$-approximate equilibrium.
This algorithm follows from the existence of simple approximate equilibria of anonymous games,
where each player plays one strategy with probability $1-\delta$, for some small $\delta$,
and plays uniformly at random with probability $\delta$.

Our approach exploits the connection between Nash equilibria in anonymous games
and Poisson multinomial distributions (PMDs). Specifically, we prove a new probabilistic lemma
establishing the following: Two PMDs,  with large variance in each direction, whose first few moments
are approximately matching are close in total variation distance. Our structural result strengthens previous
work by providing a smooth tradeoff between the variance bound and the number of matching moments.
\end{abstract}
\end{titlepage}

\section{Introduction}
Anonymous games are multiplayer games
in which the utility of each player depends on her own strategy,
as well as the number (as opposed to the identity) of other
players who play each of the strategies.
Anonymous games comprise an important class of succinct games
--- well-studied in the economics literature (see, e.g., \cite{Milchtaich96, Blonski99, Blonski05}) ---
capturing a wide range of phenomena that frequently arise in practice,
including congestion games, voting systems, and auctions.

In recent years, anonymous games have attracted significant attention in
TCS~\cite{DaskalakisP07, DaskalakisP08, DaskalakisP09, DaskalakisP15, goldbergT15, ChenDO15, DaskalakisDKT15, DiakonikolasKS16stoc},
with a focus on understanding the computational complexity of their
(approximate) Nash equilibria. Consider the family of anonymous games where the number of players,
$n$, is large and the number of strategies, $k$, is bounded. It was recently shown by Chen {\em et al.}~\cite{ChenDO15} that
computing an $\eps$-approximate Nash equilibrium of these games is PPAD-Complete when $\eps$ is exponentially small,
even for anonymous games with $5$ strategies\footnote{\cite{ChenDO15} showed that computing
an equilibrium of $7$-strategy anonymous games is PPAD-Complete, but $3$ of the $7$ strategies
in their construction can be merged, resulting in a $5$-strategy anonymous game.}.

On the algorithmic side,
Daskalakis and Papadimitriou~\cite{DaskalakisP07, DaskalakisP08}
presented the first polynomial-time approximation scheme (PTAS) for this problem with running time $n^{(1/\eps)^{\Omega(k)}}$.
For the case of $2$-strategies, this bound was improved~\cite{DaskalakisP09, DDS12, DaskalakisP15}
to $\poly(n)\cdot(1/\eps)^{O(\log^2 (1/\eps))}$, and subsequently sharpened to $\poly(n)\cdot(1/\eps)^{O(\log(1/\eps))}$ in~\cite{DKS16-siirv-colt}).

In recent work, Daskalakis {\em et al.}~\cite{DaskalakisDKT15} and Diakonikolas {\em et al.}~\cite{DiakonikolasKS16stoc}
generalized the aforementioned results~\cite{DaskalakisP15, DKS16-siirv-colt} to any fixed number $k$ of strategies,
obtaining algorithms for computing $\eps$-well-supported equilibria
with runtime of the form $n^{\poly(k)} \cdot (1/\eps)^{k \log(1/\eps)^{O(k)}}$.
That is, the problem of computing approximate Nash equilibria in anonymous games with a fixed number of strategies
admits an {\em efficient} polynomial-time approximation scheme (EPTAS). Moreover, the dependence of the running time
on the parameter $1/\eps$ is {\em quasi-polynomial} -- as opposed to exponential.

We note that all the aforementioned algorithmic results are obtained by exploiting a connection
between Nash equilibria in anonymous games and Poisson multinomial distributions (PMDs).
This connection -- formalized in~\cite{DaskalakisP07, DaskalakisP08} --
translates constructive upper bounds on $\eps$-covers for PMDs to upper bounds on computing
$\eps$-Nash equilibria in anonymous games (see Section~\ref{sec:background} for formal definitions).
Unfortunately, as shown in~\cite{DaskalakisDKT15, DiakonikolasKS16stoc}, this
``cover-based'' approach cannot lead to qualitatively faster algorithms, due to a matching existential
lower bound on the size of the corresponding $\eps$-covers.
In a related algorithmic work, Goldberg and Turchetta \cite{goldbergT15} studied two-strategy anonymous games
  ($k=2$) and designed a polynomial-time algorithm that computes an $\eps$-approximate Nash equilibria for $\eps = \Omega(n^{-1/4})$.

The aforementioned discussion prompts the following natural
question: {\em What is the precise approximability of computing Nash equilibria in anonymous games?}
In this paper, we make progress on this question by establishing the following result:
For any $\delta>0$, and any $n$-player anonymous game with a constant number of strategies,
there exists a $\poly_{\delta}(n)$ time algorithm
that computes an $\eps$-approximate Nash equilibrium of the game, for $\eps = 1/n^{1-\delta}$
\footnote{The runtime of our algorithm depends exponentially in $1/\delta$.
We remind the reader that the algorithms of~\cite{DaskalakisDKT15, DiakonikolasKS16stoc}
run in quasi-polynomial time for any value of $\eps$ inverse polynomial in $n$.}.
Moreover, we show that the existence of a polynomial-time algorithm
that computes an $\eps$-approximate Nash equilibrium for $\eps = 1/n^{1+\delta}$,
for any small constant $\delta>0$ -- i.e., slightly better than the approximation guarantee of our algorithm
-- would imply the existence of a fully polynomial-time approximation scheme (FPTAS) for the problem.
That is, we essentially show that the value $\eps = 1/n$ is the threshold
for the polynomial-time approximability of Nash equilibria in anonymous games, unless there is an FPTAS.
In the following subsection, we describe our results in detail and provide an overview of our techniques.
  
\subsection{Our Results and Techniques} \label{ssec:results}
We study the following question:
\begin{quote}
\em
For $n$-player $k$-strategy anonymous games,
how small can $\eps$ be \mbox{(as a function of $n$)},
so that an $\eps$-approximate Nash equilibrium can be computed in polynomial time?
\end{quote}

\paragraph{Upper Bounds.}
We present two different algorithms (Theorems \ref{thm:main} and \ref{thm:simple})
for computing approximate Nash equilibria in anonymous games.
Both algorithms run in polynomial time and compute $\eps$-approximate equilibria for an inverse polynomial
$\eps$ above a certain threshold.

\begin{theorem}[Main]
\label{thm:main}
For any $\delta>0$, and any $n$-player $k$-strategy anonymous game,
there is a $\poly_{\delta, k}(n)$ time algorithm that computes an $(1/n^{1-\delta})$-approximate
equilibrium of the game.
\end{theorem}

\begin{theorem}
\label{thm:simple}
For any $n$-player $k$-strategy anonymous game,
we can compute an $\tildeO{n^{-1/3} k^{11/3}}$-approximate equilibrium
in time $\tildeO{(n+k) k n^k}$ .
\end{theorem}

Prior to our work, for $k>2$, no polynomial time $\eps$-approximation was known for any inverse polynomial $\eps$.
For $k=2$, the best previous result is due to \cite{goldbergT15} who gave
a polynomial-time algorithm for $\eps = \Omega(n^{-1/4})$.

\medskip

\noindent {\bf \em Overview of Techniques.}
The high-level idea of our approach is this: If the desired accuracy $\eps$ is above a certain threshold,
we do not need to enumerate over an $\eps$-cover for the set of all PMDs.
Our approach is in part inspired by \cite{goldbergT15},
who design an algorithm (for $k=2$ and $\eps = \Omega(n^{-1/4})$)
in which all players use one of the two pre-selected mixed strategies.
We note that for $k=2$, PMDs are tantamount to Poisson Binomial distributions (PBDs), i.e., sums
of independent Bernoulli random variables.
The \cite{goldbergT15} algorithm can be equivalently interpreted as guessing
a PBD from an appropriately small set. One reason this idea succeeds is the following:
If every player randomizes, then the variance of the resulting PBD must be relatively high,
and (as a result) the corresponding subset of PBDs has a smaller cover.

Our quantitative improvement for the $k=2$ case is obtained as follows:
Instead of enforcing players to selected specific mixed strategies -- as in \cite{goldbergT15} --
we show that there always exists an $\eps$-approximate equilibrium where the associated PBD
has variance at least $\bigTheta{n\eps}$.
When $\eps = n^{-c}$ for some $c < 1$, the variance is an inverse polynomial of $n$.
We then construct a polynomial-size $\eps$-cover for the subset of PBDs with variance at least this much,
which leads to a polynomial-time algorithm for computing $\eps$-approximate equilibria in $2$-strategy anonymous games.

The idea for the general case of $k>2$ is similar, but the details are more elaborate,
since the structure of PMDs is more complicated for $k>2$.
We proceed as follows: We start by showing that there is an $\eps$-approximate equilibrium whose corresponding PMD has a large variance in each direction.
Our main structural result is a robust moment-matching lemma (Lemma \ref{lem:pmd-fewer}),
which states that the closeness in low-degree moments of two PMDs, with large variance in each direction,
implies their closeness in total variation distance. The proof of this lemma uses Fourier analytic techniques,
building on and strengthening previous work~\cite{DiakonikolasKS16stoc}.
As a consequence of our moment-matching lemma,
we can construct a polynomial-size $(\eps/5)$-cover for PMDs with such large variance.
We then iterate through this cover to find an $\eps$-approximate equilibrium,
using a dynamic programming approach similar to the one in \cite{DaskalakisP15}.

We now provide a brief intuition of our moment-matching lemma.
Intuitively, if the two PMDs in question are both very close to discrete Gaussians,
then the closeness in the first \emph{two} moments is sufficient.
Lemma \ref{lem:pmd-fewer} can be viewed as a generalization of this intuition,
which gives a quantitative tradeoff between
the number of moments we need to approximately match and the size of the variance.
The proof of Lemma \ref{lem:pmd-fewer} exploits the sparsity of the Fourier transform of our PMDs,
and the fact that higher variance allows us to take fewer terms in the Taylor expansion
when we use moments to approximate the logarithmic Fourier transform.
This completes the proof sketch of Theorem~\ref{thm:main}.

Our second algorithm (Theorem \ref{thm:simple}) addresses the need to play simple strategies.
Players tend to favor simple strategies which are easier to learn and implement,
even if these strategies might have slightly sub-optimal payoffs~\cite{Simon82}.
In addition, our algorithm is significantly faster in this case.
We build on the idea of \cite{goldbergT15} to ``smooth'' an anonymous game by forcing all the players to randomize.
We prove that the perturbed game is Lipschitz and therefore admits a pure Nash equilibrium,
which corresponds to simple approximate equilibria of a specific form in the original game:
Each player plays one strategy with probability $1-\delta$ for some small $\delta$,
and plays other strategies uniformly at random with probability $\delta$.
To prove that the perturbed game is Lipschitz, we make essential use of the recently established multivariate central limit theorem (CLT)
in Daskalakis {\em et al.}~\cite{DaskalakisDKT15} and Diakonikolas, Kane and Stewart \cite{DiakonikolasKS16stoc}
to show that if we add a little more noise (corresponding to $\delta = \Theta(n^{-1/3})$),
the associated PMD is sufficiently close to a discrete Gaussian.

\paragraph{Lower Bounds.}
When $\eps = 1/n$, we can show that there is an $\eps$-approximate equilibrium
where the associated PMD has a variance at least $1/k$ in every direction.
Unfortunately, the PMDs in the explicit quasi-polynomial-size lower bounds
given in \cite{DaskalakisDKT15, DiakonikolasKS16stoc} satisfy this property.
Thus, we need a different approach to get a polynomial-time algorithm for $\eps = 1/n$ or smaller.

In fact, we prove the following results, which states that
even a slight improvement of our upper bound in Theorem~\ref{thm:main}
would imply an FPTAS for computing Nash equilibria in anonymous games.
It is important to note that Theorem \ref{thm:fptas} applies to all algorithms,
not only the ones that leverage the structure of PMDs.

\begin{theorem} \label{thm:fptas}
For $n$-player $k$-strategy anonymous games with $k = \bigO{1}$,
if we can compute an $\bigO{n^{-c}}$-approximate equilibrium in polynomial time for some constant $c > 1$,
then there is an FPTAS \footnote{A fully polynomial-time approximation scheme (FPTAS)
is an algorithm that runs in time $\poly(n, 1/\eps)$ and returns
an $\eps$-optimal solution, or in our context, returns an $\eps$-approximate Nash equilibrium.}
for computing (well-supported) Nash equilibria of $k$-strategy anonymous games.
\end{theorem}

\noindent {\bf Remark.}
As observed in \cite{DaskalakisDKT15},
because there is a quasi-polynomial time algorithm for computing an
$(n^{-c})$-approximate equilibrium in anonymous games,
the problem cannot be PPAD-Complete unless \texttt{PPAD} $\subseteq$ \texttt{Quasi-PTIME}.
On the other hand, we do not know how to improve the quasi-polynomial-time upper bounds of \cite{DaskalakisDKT15, DiakonikolasKS16stoc}
when $\eps < 1/n$.

Recall that computing an $\eps$-approximate equilibrium of a two-player general-sum $n \times n$ game
(2-NASH) for constant $\eps$ also admits a quasi-polynomial-time algorithm~\cite{LiptonMM03}.
Very recently, Rubinstein~\cite{Rubinstein16a} showed that, assuming the exponential time hypothesis (ETH) for PPAD,
for some sufficiently small universal constant $\eps > 0$,
quasi-polynomial-time is necessary to compute an $\eps$-approximate equilibrium
of 2-NASH.
It is a plausible conjecture that
quasi-polynomial-time is also required for $\eps$-Nash equilibria in anonymous games,
when $\eps = n^{-c}$ for some constant $c > 1$.
In particular, this would imply that there is no FPTAS
for computing approximate Nash equilibria in anonymous games,
and consequently the upper bound of Theorem~\ref{thm:main} is essentially tight.

\section{Notation and Background} \label{sec:background}
\paragraph{Anonymous Games.}
We study anonymous games $(n, k, \{u^i_a\}_{i\in[n], a\in[k]})$
  with $n$ players labeled by $[n] = \{1, \ldots, n\}$,
  and $k$ common strategies labeled by $[k]$ for each player.
The payoff of a player depends on her own strategy,
  and how many of her peers choose which strategy,
  but not on their identities.
When player $i \in [n]$ plays strategy $a \in [k]$,
  her payoffs are given by a function $u^i_a$ that maps the possible outcomes
  (partitions of all other players) $\Pi^k_{n-1}$ to the interval $[0, 1]$,
  where $\Pi^k_{n-1} = \{(x_1, \ldots, x_k) \mid x_j \in \N \wedge \sum_{j=1}^k x_j = n-1 \}$.

\paragraph{Approximate Equilibria.}
We denote by $\Delta_S$ a distribution on the set $S$.
A \emph{mixed strategy} is an element of $\Delta_{[k]}$,
  and a \emph{mixed strategy profile} $\vec{s} = (\vec{s}_1, \ldots, \vec{s}_n)$
  maps every player $i$ to her mixed strategy $\vec{s}_i \in \Delta_{[k]}$.
We use $\vec{s}_{-i}$
  to denote the strategies of players other than $i$ in $\vec{s}$.

A mixed strategy profile $\vec{s}$ is an \emph{$\eps$-approximate Nash equilibrium}
  for some $\eps \ge 0$ iff
\[ \forall i \in [n], \forall a' \in [k], \quad
  \expectOver{x \sim \vec{s}_{-i}}{u^i_{a'}(x)} \le \expectOver{x \sim \vec{s}_{-i}, a \sim \vec{s}_i}{u^i_a(x)} + \eps, \]
where $x \in \Pi^k_{n-1}$ is the partition formed by $n-1$ random samples (independently) drawn from $[k]$
  according to the distributions $\vec{s}_{-i}$.
Note that given a mixed strategy profile $\vec{s}$, we can compute a player's expected payoff
  to precision $\eps$ in time $\poly(n^k \log(1/\eps))$ by straightforward dynamic programming,
  and hence throughout this paper we assume that we can compute players' payoffs exactly given their mixed strategies.

\paragraph{Poisson Multinomial Distributions.}
A \emph{$k$-Categorical Random Variable} ($k$-CRV) is a vector random variable
  supported on the set of $k$-dimensional basis vectors $\{e_1, \ldots, e_k\}$.
A $k$-CRV is $i$-maximal if $e_i$ is its most likely outcome (break ties by taking the smallest index $i$).
A $k$-Poisson Multinomial Distribution of order $n$, or an $(n, k)$-PMD,
is a vector random variable of the form $X = \sum_{i=1}^n X_i$ where the $X_i$'s are independent $k$-CRVs.
The case of $k=2$ is usually referred to as Poisson Binomial Distribution (PBD).

Note that a mixed strategy profile $\vec{s} = (\vec{s}_1, \ldots, \vec{s}_n)$
  of an $n$-player $k$-strategy anonymous game corresponds to the $k$-CRVs $\{X_1, \ldots, X_n\}$
  where $\Pr[X_i = e_a] = \vec{s}_i(a)$.
The expected payoff of player $i \in [n]$ for playing pure strategy $a \in [k]$ can also be written as $\expect{u^i_a (X_{-i})} = \expect{u^i_a \left(\sum_{j\neq i, j \in [n]} X_j\right)}$.

Let $X = \sum_{i=1}^n X_i$ be an $(n, k)$-PMD such that for $i \in [n]$ and $j \in [k]$ we denote $p_{i,j} = \Pr[X_i = e_j]$, where $\sum_{j=1}^k p_{i,j} = 1$.
For $m = (m_1, \ldots, m_k) \in \Z_+^k$, we define the $m^{th}$-parameter moments of $X$ to be $M_m(X) \eqdef \sum_{i=1}^n \prod_{j=1}^k p_{i,j}^{m_j}$.
We refer to $\normone{m} = \sum_{j=1}^k m_j$ as the \emph{degree} of the parameter moment $M_m(X)$.

\paragraph{Total Variation Distance and Covers.}
The total variation distance between two distributions $P$ and $Q$
  supported on a finite domain $A$ is
\[ \totalvardist{P}{Q} := \max_{S \subseteq A}\abs{P(S) - Q(S)} = (1/2) \cdot \normone{P-Q}. \]
If $X$ and $Y$ are two random variables ranging over a finite set,
  their total variation distance $\totalvardist{X}{Y}$ is defined as
  the total variation distance between their distributions.
For convenience, we will often blur the distinction between a random variable and its distribution.

Let $(\mathcal{X},d)$ be a metric space.
Given $\eps > 0$, a subset $\mathcal{Y} \subseteq \mathcal{X}$ is said to be a proper $\eps$-cover
  of $\mathcal{X}$ with respect to the metric $d: \mathcal{X}^2 \rightarrow \R_+$,
  if for every $X \in \mathcal{X}$ there exists some $Y \in \mathcal{Y}$ such that $d(X, Y) \le \eps$.
In this work, we will be interested in constructing $\eps$-covers
  for high-variance PMDs under the total variation distance metric.

\paragraph{Multidimensional Fourier Transform.}
For $x \in \R$, we will denote $e(x) \eqdef \exp(-2 \pi i x)$.
The (continuous) Fourier Transform of a function $F: \Z \rightarrow \C$ is the function
  $\wh{F}: [0, 1]^k \rightarrow \C$ defined as $\wh{F}(\xi) = \sum_{x\in \Z_k} e(\xi \cdot x) F(x)$.
For the case that $F$ is a probability mass function, we can equivalently write $\wh{F}(\xi) = \expectOver{x \sim F}{e(\xi \cdot x)}$.

Let $X = \sum_{i=1}^n X_i$ be an $(n, k)$-PMD with $p_{i,j} \eqdef \Pr[X_i = e_j]$.
To avoid clutter in the notation, we will sometimes use the symbol $X$ to denote the corresponding probability mass function.
With this convention, we can write that $\wh{X}(\xi) = \prod_{i=1}^n \wh{X_i}(\xi) = \prod_{i=1}^n \sum_{j=1}^k e(\xi_j) p_{i,j}$.

\section{Searching Fewer Moments: Proof of Theorem~\ref{thm:main}}
In this section, we present a polynomial-time algorithm that,
for $n$-player anonymous games with bounded number of strategies,
computes an $\eps$-approximate equilibrium with $\eps = n^{-c}$ for any constant $c < 1$.
As a warm up, we start by describing the simpler setting of two-strategy anonymous games ($k = 2$).
The main results of this section is Theorem \ref{thm:main}
that applies to general $k$-strategy anonymous games for any constant $k \geq 2$.

At a high level, we first prove the existence of $\eps$-approximate Nash equilibria
in which the corresponding PMDs have high variance and every player randomizes (Lemma \ref{lem:perturb}).
We then use our robust moment matching lemma (Lemma \ref{lem:pmd-fewer})
to show that when two PMDs have high variances,
the closeness in their constant-degree parameter moments implies their closeness in total variation distance.
The fact that matching the constant-degree moments suffices
allows us to construct a polynomial-size $(\eps/5)$-cover for set subset of all PMDs with large variance.
We then iterate through this cover to find an $\eps$-approximate equilibrium (Algorithm \ref{alg:main}).

\begin{lemma}
\label{lem:perturb}
For an $n$-player $k$-strategy anonymous game,
  there always exists an $\eps$-approximate equilibrium
  where every player plays each strategy with probability at least $\frac{\eps}{k-1}$.
\end{lemma}
\begin{proof}
Given an anonymous game $G = (n, k, \{u^i_a\}_{i\in[n], a\in[k]})$,
we smooth players' utility functions by requiring every player to randomize.
Fix $\eps > 0$, we define an $\eps$-perturbed game $G_\eps$ as follows.
When a player plays some pure strategy $a \in [k]$ in $G_\eps$,
we map it back to the original game as if she plays strategy $j$ with probability $1-\eps$,
and plays some other strategy $a' \neq a$ uniformly at random
(i.e., she plays $a'$ with probability $\frac{\eps}{k-1}$).
Her payoff in $G_\eps$ also accounts for such perturbation,
and is defined to be her expected payoff given that all the players (including herself)
would deviate to other strategies uniformly at random with probability $\eps$.

Formally, let $X_\eps(e_j)$ denote the $k$-CRV that takes value $e_j$ with probability $1-\eps$,
and takes value $e_{j'}$ with probability $\frac{\eps}{k-1}$ for each $j' \neq j$.
The payoff structure of $G_\eps$ is given by
\[ u'^i_a(x) := (1-\eps)\expect{u^i_{a}(M_\eps(x))} + \frac{\eps}{k-1} \sum_{a' \neq a}\expect{u^i_{a'}(M_\eps(x))}, \quad \forall i\in [n], a\in [k], x \in \Pi^k_{n-1}, \]
where $M_\eps(x) = \sum_{j\in[k]} x_j X_\eps(e_j)$ is an $(n-1,k)$-PMD that corresponds to the
perturbed outcome of the partition $x \in \Pi^k_{n-1}$ of all other players.

Let $\vec{s}' = (\vec{s}'_1, \ldots, \vec{s}'_n)$ denote any exact Nash equilibrium of $G_\eps$.
We can interpret this mixed strategy profile in $G$
equivalently as $\vec{s} = (\vec{s}_1, \ldots, \vec{s}_n)$,
where $\vec{s}_i = (1-\frac{k\eps}{k-1}) \vec{s}'_i + \frac{\eps}{k-1} \vec{1}$, where $\vec{1} = (1, \ldots, 1)$.
We know that under $\vec{s}$ each player has no incentive to
deviate to the mixed strategies $X_\eps(e_j)$ for all $j \in [k]$,
therefore a player can gain at most $\eps$ by deviating to pure strategies in $G$,
so $\vec{s}$ is an $\eps$-approximate equilibrium with
$\vec{s}_i(j) \ge \frac{\eps}{k-1}$ for all $i \in [n]$, $j \in [k]$.
\end{proof}

\paragraph{Warm-up: The Case of $k=2$ Strategies.}
\label{sec:fewer-moments-pbd}
For two-strategy anonymous games ($k=2$),
if all the players put at least $\eps$ probability mass on both strategies,
the resulting PBD is going to have variance at least $n \eps (1-\eps)$.
When $\eps = n^{-c}$ for some constant $c < 1$, the variance is
at least $\bigTheta{n^{1-c}} = n^{\Theta(1)}$.
We can now use the following lemma from \cite{DiakonikolasKS16coltPBD},
which states that if two PBDs $P$ and $Q$ are close in the first few moments,
then $P$ and $Q$ are $\eps$-close in total variation distance.
Note that without any assumption on the variance of the PBDs,
we would need to check the first $\bigO{\log(1/\eps)}$ moments,
but when the variance is $n^{\Omega(1)}$, which is the case in our application,
we only need the first constant number of moments to match.

\begin{lemma}[\cite{DiakonikolasKS16coltPBD}]
\label{lem:pbd_fewer}
Let $\eps > 0$.
Let $P$ and $Q$ be $n$-PBDs with $P$ having parameters
$p_1, \ldots, p_{s} \le 1/2$ and $p'_1, \ldots, p'_{s'} > 1/2$,
and $Q$ having parameters $q_1, \ldots, q_{s} \le 1/2$ and $q'_1, \ldots, q'_{s'} > 1/2$.
Suppose that $V = \var[P]+1 = \bigTheta{\var[Q]+1}$ and
let $C > 0$ be a sufficiently large constant.
Suppose furthermore that for $A = C \sqrt{\log(1/\eps)/V}$ and
for all positive integers $\ell$ it holds
\begin{align}
  A^\ell \left( \abs{\sum_{i=1}^s p_i^\ell - \sum_{i=1}^s q_i^\ell}
  + \abs{\sum_{i=1}^{s'} (1-p'_i)^\ell - \sum_{i=1}^{s'} (1-q'_i)^\ell} \right)
  < \frac{\eps}{C \log(1/\eps)}
  \label{eqn:pbd_fewer}
\end{align}
Then $\totalvardist{P}{Q} < \eps$.
\end{lemma}

Let $\eps = n^{-c}$.
For Lemma \ref{lem:pbd_fewer} we have $V \ge n \eps(1-\eps)$ and
$A = \bigTheta{\sqrt{\log(1/\eps) / V}} = \bigO{\sqrt{\frac{\log n}{n^{1-c}}}}$.
The difference in the moments of parameters of $P$ and $Q$ in Equation (\ref{eqn:pbd_fewer})
is bounded from above by $n$, so whenever $\ell > \frac{2+2c}{1-c}$,
the condition in Lemma \ref{lem:pbd_fewer} is automatically satisfied for sufficiently large $n$ because
\[ A^\ell n = \bigO{\frac{\log^{\ell/2} n}{n^{(1-c) \ell / 2}} n} <
  {\frac{1}{C \cdot n^c \cdot c \log n}} = \frac{\eps}{C \log(1/\eps)}. \]
So it is enough to search over the first $\ell = \bigTheta{\frac{1}{1-c}}$ moments
when each player put probability at least $\bigOmega{n^{-c}}$ on both strategies.
The algorithm for finding such an $\eps$-approximate equilibrium uses moment search
and dynamic programming, and is given for the case of general $k$ in the remainder of this section.

\paragraph{The General Case: $k$ Strategies.}
We now present our algorithm for $n$-player anonymous games with $k > 2$ strategies and prove Theorem \ref{thm:main}.
The intuition of the $k=2$ case carries over to the general case, but the details are more elaborate.
First, we show (Claim \ref{clm:all-randomize-nash}) that there exists an $\eps$-approximate equilibrium
whose corresponding PMD has variance $(n\eps / k)$ in all directions orthogonal to the vector $\mathbf{1} = (1, \ldots, 1)$.
Then, we prove (Lemma \ref{lem:pmd-fewer}) that when two PMDs have such high variances,
the closeness in their constant-degree parameter moments
translates to their closeness in total variation distance.
This structural result allows us to build a polynomial-size
cover for all PMDs with high variance,
which leads to a polynomial-time algorithm for computing $\eps$-approximate Nash equilibria (Algorithm \ref{alg:main}).
 
We first prove that when all the players put probability at least $\frac{\eps}{k-1}$ on each strategy,
  the covariance matrix of the resulting PMD has relatively large eigenvalues,
  except the zero eigenvalue associated with the all-one eigenvector.
The all-one eigenvector has eigenvalue zero because the coordinates of $X$ always sum to $n$.
\begin{claim}
\label{clm:all-randomize-nash}
Let $X = \sum_{i=1}^n X_i$ be an $(n,k)$-PMD and let $\Sigma$ be the covariance matrix of $X$.
If $p_{i,j} = \Pr[X_i = e_j] \ge \frac{\eps}{k-1}$ for all $i \in [n]$ and $j \in [k]$,
then all eigenvalues of $\Sigma$ but one are at least $\frac{n \eps}{k-1}$.
\end{claim}
\begin{proof}
For any unit vector $v \in \R^k$ that is orthogonal to the all-one vector $\mathbf{1}$,
i.e., $\sum_j {v_j} = 0$ and $\sum_j {v_j^2}=1$,
combining this with the assumption that $p_{i,j} \ge \frac{\eps}{k-1}$ we have,
\begin{align*}
\var[v^T X_i] &= \expect{\left(v^T X_i - \expect{(v^T X_i)}\right)^2} \\
  &= \sum_{j=1}^n p_{i,j} \left(v_j - \sum_{j'=1}^n p_{i,j'} v_{j'} \right)^2 \\
  &\ge \min_j \{p_{i,j}\} \cdot \sum_{j=1}^n\left( v_j^2 + \left(\sum_{j'=1}^n p_{i,j'} v_{j'}\right)^2 - 2 v_j \left(\sum_{j'=1}^n p_{i,j'} v_{j'} \right)\right) \\
  &= \min_j \{p_{i,j}\} \cdot \left(1 + n \left(\sum_{j'=1}^n p_{i,j'} v_{j'}\right)^2 \right) \\
  &\ge \frac{\eps}{k-1}.
  \end{align*}
Therefore,
\[ v^T \Sigma v = \var[v^T X] = \sum_{i=1}^n \var[v^T X_i] \ge \frac{n\eps}{k-1}.\]
So, for all eigenvectors $v$ orthogonal to $\vec{1}$,
we have $v^T \Sigma v = \lambda v^T v = \lambda \ge \frac{n\eps}{k-1}$ as claimed.
\end{proof}

The following robust moment-matching lemma provides a bound on how close degree $\ell$ moments need to be
so that two $(n,k)$-PMDs are $\eps$-close to each other,
under the assumption that $n \gg k$ (the anonymous game has many players and few strategies)
and $p_{i,j} \ge \frac{\eps}{k-1}$ (every player randomizes).
Lemma \ref{lem:pmd-fewer} allows us to build a polynomial-size $(\eps/5)$-cover
for PMDs with high variance, and since we know that there is an $\eps$-approximate equilibrium
with a high variance, we are guaranteed to find one in our cover.
\begin{lemma}
\label{lem:pmd-fewer}
Fix $0 < c < 1$ and let $\eps = n^{-c}$.
Assume that $n = k^{\Omega(k)}$ for some sufficiently large constant in the exponent.
Let $X$, $Y$ be $(n,k)$-PMDs with $X=\sum_{i=1}^k X^i$, $Y=\sum_{i=1}^k Y^i$
  where each $X^i$, $Y^i$ is an $i$-maximal PMD.
Let $\Sigma$ and $\Sigma'$ denote the covariance matrices of $X$ and $Y$ respectively.
If all eigenvalues of $\Sigma, \Sigma'$ but one are at least $\eps n/k$,
  and for $\ell \le \frac{2+2c}{1-c}$ all the parameter moments $m$ of degree $\ell$ satisfy that
\[ \abs{M_m(X^i)-M_m(Y^i)} \le n^{-c}. \]
Then, we have that $\totalvardist{X}{Y} \le \eps$.
\end{lemma}

Lemma~\ref{lem:pmd-fewer} follows from the next proposition whose proof
is given in the following subsection.

\begin{proposition}
\label{prop:moments-variance-tradeoff}
Let $\eps > 0$.
Let $X$, $Y$ be $(n,k)$-PMDs with $X=\sum_{i=1}^k X^i$, $Y=\sum_{i=1}^k Y^i$
where each $X^i$, $Y^i$ is an $i$-maximal PMD.
Let $\Sigma$ and $\Sigma'$ denote the covariance matrices of $X$ and $Y$ respectively,
where all eigenvalues of $\Sigma$ and $\Sigma'$ but one
are at least $\sigma^2$, where $\sigma \geq \poly(k \log(1/\eps))$.
Suppose that for $1 \leq i \leq k$, $\ell \ge 1$, for all moments $m$ of degree $\ell$ with $m_i=0$, we have that
\begin{align}
\abs{M_m(X^i)-M_m(Y^i)} \leq \frac{\eps \cdot \sigma^\ell}{C'^{k+\ell} \cdot k^{3\ell/2+1} \cdot \log^{k+\ell/2}(1/\eps)}
\end{align}
for a sufficiently large constant $C'$. Then $\dtv(X,Y) \leq \eps$.
\end{proposition}

The proof of Proposition \ref{prop:moments-variance-tradeoff} exploits the sparsity of the continuous Fourier transform
of our PMDs, as well as careful Taylor approximations of the logarithm of the Fourier transform.

\begin{proof}[Proof of Lemma~\ref{lem:pmd-fewer} from Proposition~\ref{prop:moments-variance-tradeoff}]
In order to guarantee that $\totalvardist{X}{Y} \le \eps$,
Proposition \ref{prop:moments-variance-tradeoff} requires the following condition to hold for a sufficiently large constant $C'$:
\begin{align}
\label{eqn:pmd-moments-variance}
\abs{M_m(X^i)-M_m(Y^i)} \leq \frac{\eps}{k (C'\log(1/\eps))^k}\cdot \left(\frac{\sqrt{\eps n/k}}{C' k^{3/2} \log^{1/2}(1/\eps)}\right)^\ell, \quad \forall i \in [k], \ell \ge 1.
\end{align}
To prove the lemma, we use the fact that $n \gg k$ and essentially ignore all the terms except polynomials of $n$.
Formally, we first need to show that
\[ \frac{\eps}{k (C'\log(1/\eps))^k}\cdot \left(\frac{\sqrt{\eps n/k}}{C' k^{3/2} \log^{1/2}(1/\eps)}\right)^\ell \ge n^{-c}, \quad \forall \ell \ge 1, \]
  under the assumption that $c < 1$, $\eps = n^{-c}$ and $n \ge k^{O(k/(1-c))}$.
After substituting $\eps = n^{-c}$,
  observe that $n^{1-c} \ge C'^2 k^4 \log n$, so the term inside the $\ell$-th power is greater than 1.
Thus, we only need to check this inequality for $\ell = 1$, which simplifies to $n^{1-c} \ge C'^{2k+2} k^6 (\log n)^{2k}$ and holds true.

In addition, we need to show that condition (\ref{eqn:pmd-moments-variance}) holds automatically for $\ell > \frac{2+2c}{1-c}$.
This follows from the fact that the difference in parameter moments is at most $n$ and $n \gg k$,
\[ \abs{M_m(X^i)-M_m(Y^i)} \le n \le \frac{\eps}{k (C'\log(1/\eps))^k}\cdot \left(\frac{\sqrt{\eps n/k}}{C' k^{3/2} \log^{1/2}(1/\eps)}\right)^\ell, \quad \forall \ell > \frac{2+2c}{1-c}. \qedhere \]
\end{proof}

We recall some of the notations for readability before we describe the construction of our $\eps$-cover of high-variance PMDs.
We use $X$ to denote a generic $(\ell,k)$-PMD for some $\ell \in [n]$, and we denote $p_{i,j} = \Pr[X_i = e_j]$.
We use $A_t \subseteq [\ell]$ to denote the set of $t$-maximal CRVs in $X$,
  where a $k$-CRV is $t$-maximal if $e_t$ is its most likely outcome,
  and we use $X^t = \sum_{i \in A_t} X_i$ to denote the $t$-maximal component PMD of $X$.
For a vector $m = (m_1, \ldots, m_k) \in \Z_+^k$, we define $m^{th}$ parameter moment of $X^t$ to be
  $M_m(X^t) = \sum_{i\in A_t} \prod_{j=1}^k p_{i,j}^{m_j}.$
We refer to $\normone{m} = \sum_{j=1}^k m_j$ as the \emph{degree} of $M_m(X)$.
We use $\cals$ to denote the set of all $k$-CRVs whose probabilities are multiples of $\frac{\eps}{20kn}$.

Lemma \ref{lem:pmd-fewer} states that the high-degree parameter moments match automatically,
  which allows us to impose an appropriate grid on the low-degree moments
  to cover the set of high-variance PMDs.
The size of this cover can be bounded by a simple counting argument:
We have at most $k^{O(\frac{1}{1-c})}$ moments with degree at most $O(\frac{1}{1-c})$,
  and we need to approximate these moments for each $t$-maximal component PMDs,
  so there are at most $k \cdot k^{O(\frac{1}{1-c})} = k^{O(\frac{1}{1-c})}$ moments $M_m(X^t)$ that we care about.
We approximate these moments to precision $n^{-c}$, and the moments are at most $n$, so the size of the cover is $\left(\frac{n}{n^{-c}}\right)^{k^{O(\frac{1}{1-c})}} = n^{k^{O(1/1-c)}}$.

We define this grid on low-degree moments formally in the following lemma.
For every $(\ell,k)$-PMD $X$ with $\ell \in [n]$, we associate some \emph{data} $D(X)$ with $X$,
  which is a vector of the approximate values of the low-degree moments $M_m(X^t)$ of $X$.
\begin{lemma}
\label{lem:pmd-data}
Fix $0 < c < 1$ and $n$, let $\eps = n^{-c}$.
We define the data $D(W)$ of a $k$-CRV $W$ as:
\[ D(W)_{m,t} =
  \begin{cases}
    M_m(W) \text{ rounded to the nearest integer multiple of } n^{-c} / n, & \text{ if $W$ is $t$-maximal.} \\
    0, &\text{otherwise.}
  \end{cases} \]
For $\ell \in [n]$, we define the data of an $(\ell,k)$-PMD
  $X = \sum_{i=1}^\ell X_i$ to be the sum of the data of its $k$-CRVs:
  $D(X) = \sum_{i=1}^\ell D(X_i)$.
The data $D(X)$ satisfies two important properties:
\begin{enumerate}
\item (Representative) If $D(X) = D(Y)$ for two $(n, k)$-PMDs or two $(n-1, k)$-PMDs, then $\totalvardist{X}{Y} \le \eps$.
\item (Extensible) For independent PMDs $X$ and $Y$, we have that $D(X+Y) = D(X) + D(Y)$.
\end{enumerate}
\end{lemma}
\begin{proof}
The extensible property follows directly from the definition of $D(X)$.
To see the representative property,
  note that we round $M_m(W)$ to the nearest integer multiple of $n^{-c} / n$,
  so the error in the moments of $W$ is at most $n^{-c} / (2n)$.
When we add up the data of an $(n,k)$-PMD or $(n-1,k)$-PMD,
  the error in the moments of each $t$-maximal component PMDs is at most $n^{-c}/2$.
So if two PMDs $X$ and $Y$ have the same data, their low-degree moments differ by at most $n^{-c}$,
  and then by Lemma \ref{lem:pmd-fewer} we have $\totalvardist{X}{Y} \le \eps$.
\end{proof}

\begin{algorithm}[!ht]
  \caption{GenerateData}
  \label{alg:gendata}
  \SetKwInOut{Input}{Input}
  \SetKwInOut{Output}{Output}
  \Input{$\{\cals_i\}_{i=1}^n$, $\eps > 0$.}
  \Output{The set of all possible data $\cald$ of $(n,k)$-PMDs $X = \sum_{i=1}^n X_i$ where $X_i \in \cals_i$.}
  $\cald_0 = \{ \vec{0} \}$\;
  \For{$\ell = 1 \ldots n$}{
   \ForAll{$D \in \cald_{\ell-1}$}{
    \ForAll{$W \in \cals_\ell$}{
     Add $D + D(W)$ to $\cald_\ell$ if it is not in $\cald_\ell$ already\;
     Keep track of an $(\ell, k)$-PMD whose data is $D + D(W)$\;
    }
   }
  }
  \Return{$\cald = \cald_n$}\;
\end{algorithm}

\begin{algorithm}[!ht]
  \caption{Moment Search}
  \label{alg:main}
  \SetKwInOut{Input}{Input}
  \SetKwInOut{Output}{Output}
  \Input{An $n$-player $k$-strategy anonymous game $G$, $\eps = n^{-c}$ for some $c < 1$.}
  \Output{An $\eps$-approximate Nash equilibrium of $G$.}
  $\cald_n$ = GenerateData(\{$\cals_i = \cals\}_{i=1}^n$, $\eps/5$)\;
  $\cald_{n-1}$ = GenerateData(\{$\cals_i = \cals\}_{i=1}^{n-1}$, $\eps/5$)\;
  \ForAll{$D \in \cald_n$}{
    Set $\cals_i = \emptyset$ for all $i$\;
    \ForAll {$X_i \in \cals$}{
      Let $D_{-i} = D - D(X_i)$\;
      \If{$\exists Y_{D_{-i}} \in \cald_{n-1}$ with $D(Y_{D_{-i}}) = D_{-i}$ and $X_i$ is a $(3\eps/5)$-best response to $Y_{D_{-i}}$}{
        Add $X_i$ to $\cals_i$\;
      }
    }
    $\cald'_n$ = GenerateData(\{$\cals_i\}_{i=1}^n$, $\eps/5$)\;
    \If{$D \in \cald'_n$}{
      \Return $(X_1, \ldots, X_n)$ with  PMD $X = $ with $D\left(\sum_{i=1}^n X_i\right) = D$ in $\cald'_n$\;
    }
  }
\end{algorithm}

Our algorithm (Algorithm \ref{alg:main}) for computing approximate equilibria is similar to the approach used in \cite{DaskalakisP15} and \cite{DiakonikolasKS16stoc}.
We start by constructing a polynomial-sized ($\eps/5$)-cover of high-variance PMDs (Algorithm \ref{alg:gendata}),
  and then iterate over this cover.
For each element in the cover, we compute the set of $(3\eps/5)$-best-responses for each player,
  and then run the cover construction algorithm again,
  but this time we only allow each player to choose from her $(3\eps/5)$-best-responses.
If we could reconstruct a PMD whose moments are close enough to the one we started with,
  then we have found an $\eps$-approximate Nash equilibrium.

Recall that a mixed strategy profile for a $k$-strategy anonymous game can be represented as
  a list of $k$-CRVs $(X_1, \ldots, X_n)$, where $X_i$ describes the mixed strategy of player $i$.
Recall that $(X_1, \ldots, X_n)$ is an $\eps$-approximate Nash equilibrium if for each player $i$ we have
  $\expect{u^i_{X_i}(X_{-i})} \ge \expect{u^i_a(X_{-i})} - \eps$ for all $a \in [k]$,
  where $X_{-i} = \sum_{j \neq i} X_j$ is the distribution of the sum of other players strategies.
\begin{lemma}
\label{lem:small-dtv-still-nash}
Fix an anonymous game $G = (n, k, \{u^i_a\}_{i\in [n],a\in [k]})$ with payoffs normalized to $[0,1]$.
Let $(X_1, \ldots, X_n)$ and $(Y_1, \ldots, Y_n)$ be two lists of $k$-CRVs.
If $X_i$ is a $\delta$-best response to $X_{-i}$,
  and $\totalvardist{X_{-i}}{Y_{-i}} \le \eps$, then $X_i$ is a $(\delta + 2\eps)$-best response to $Y_{-i}$.
Moreover, if $(X_1, \ldots, X_n)$ is a $\delta$-approximate equilibrium,
  and $\totalvardist{X_i}{Y_i} + \totalvardist{X_{-i}}{Y_{-i}} \le \eps$ for all $i \in [n]$,
  then $(Y_1, \ldots, Y_n)$ is a $(\delta + 2\eps)$-approximate equilibrium.
\end{lemma}
\begin{proof}
Since $u^i_a(x) \in [0,1]$ for all $a \in [k]$ and $x \in \Pi_{n-1}^k$, we have that
\[ \abs{\expect{u^i_a(X_{-i})} - \expect{u^i_a(Y_{-i})}} \le \totalvardist{X_{-i}}{Y_{-i}}, \quad \forall i \in [n], a \in [k]. \]
Therefore, if $\totalvardist{X_{-i}}{Y_{-i}} \le \eps$,
  and player $i$ cannot deviate and gain more than $\delta$ when other players play $X_{-i}$,
  then she cannot gain more than $(\delta+2\eps)$ when other players play $Y_{-i}$ instead of $X_{-i}$.
The second claim combines the inequality above with the fact that,
  if player $i$ plays $Y_i$ instead of $X_i$ and the mixed strategies of other players remain the same,
  her payoff changes by at most $\totalvardist{X_i}{Y_i}$.
Formally,
\[ \abs{\expect{u^i_{X_i}(Z_{-i})} - \expect{u^i_{Y_i}(Z_{-i})}} \le \totalvardist{X_i}{Y_i}, \quad \forall k\text{-CRV } X_i, Y_i, \forall (n-1,k)\text{-PMD } Z_{-i}. \qedhere \]
\end{proof}

The next lemma states that by rounding an (\eps/10)-approximate equilibrium,
  we can obtain an $(\eps/5)$-approximate equilibrium where all the probabilities are
  integer multiples of $\frac{\eps}{20 k n}$.

\begin{claim}
\label{clm:target-nash}
There is an $(\eps/5)$-approximate Nash equilibrium $(X_1, \ldots, X_n)$,
  such that for all $i \in [n]$ and $j \in [k]$,
  the probabilities $p_{i,j} = \Pr[X_i = e_j]$ are multiples of $\frac{\eps}{20 k n}$,
  and also $p_{i,j} \ge \frac{\eps}{10k}$.
\end{claim}
\begin{proof}
We start with an $(\eps/10)$-approximate Nash equilibrium $(Y_1, \ldots, Y_n)$ from Lemma \ref{lem:perturb}
  with $p_{i,j} \ge \frac{\eps}{10k}$, and then round the probabilities to integer multiples of $\frac{\eps}{10k n}$.
We construct $X_i$ from $Y_i$ as follows:
  for every $j < k$, we set $\Pr[X_i = e_j]$ to be $\Pr[Y_i = e_j]$ rounded down to a multiple of $\frac{\eps}{20 k n}$
  and we set $\Pr[X_i = e_k] = 1-\sum_{j < k} \Pr[X_i = e_j]$ so the probabilities sum to 1.
By triangle inequality of total variation distance, for every $i$ we have
  $\totalvardist{X_i}{Y_i} \le \frac{\eps}{20n}$ and $\totalvardist{X_{-i}}{Y_{-i}} \le \frac{\eps(n-1)}{20n}$.
An application of Lemma \ref{lem:small-dtv-still-nash} shows that
  $(X_1, \ldots, X_n)$ is an $(\eps/5)$-approximate equilibrium.
\end{proof}

We are now ready to prove Theorem \ref{thm:main}.
We need to show that Algorithm \ref{alg:main} always outputs an $\eps$-approximate Nash equilibrium,
  and bound the running time.
\begin{proof}[Proof of Theorem \ref{thm:main}]
We first show that the output $(X_1, \ldots, X_n)$ is an $\eps$-approximate equilibrium.
Recall that $\cals$ is the set of all $k$-CRVs whose probabilities are multiples of $\frac{\eps}{20kn}$,
  and $\cals_i \subseteq \cals$ is the set of approximate best-responses of player $i$.
When we put $X_i$ in $\cals_i$, we checked that $X_i$ is a $(3\eps/5)$-best response to $Y_{D_{-i}}$,
  note that $D(Y_{D_{-i}}) = D - D(X_i) = D(X_{-i})$,
  so by Lemma \ref{lem:pmd-data} $\totalvardist{X_{-i}}{Y_{D_{-i}}} \le \eps/5$ for all $i$.
By Lemma \ref{lem:small-dtv-still-nash}, $X_i$ is indeed an $\eps$-best response to $X_{-i}$ for all $i$.

Next we show the algorithm must always output something.
By Claim \ref{clm:target-nash} there exists an $(\eps/5)$-approximate equilibrium $X'_i$ with each $X'_i \in \cals$.
If the algorithm does not terminate successfully first, it eventually considers $D(X')$.
Because $X'_{-i}$ is an $(n-1,k)$-PMD,
  the algorithm can find some $Y_{D_{-i}}$ with $D(Y_{D_{-i}}) = D(X') - D(X'_i) = D(X'_{-i})$,
  and by Lemma \ref{lem:pmd-data} we have $\totalvardist{X'_{-i}}{Y_{D_{-i}}} \le \eps/5$ for all $i$.
Since $X'_i$ is an $(\eps/5)$-best response to $X'_{-i}$,
  Lemma \ref{lem:small-dtv-still-nash} yields that $X'_i$ is a (3\eps/5)-best response to $Y_{D_{-i}}$,
  so we would add each $X'_i$ to $\cals_i$.
Then our cover construction algorithm is guaranteed to generate a set of data that includes $D(X')$,
  and Algorithm \ref{alg:main} would produce an output.

Finally, we bound the running time of Algorithm \ref{alg:main}.
Let $N = \bigO{n^{k^{O(1/1-c)}}}$ denote the size of the $(\eps/5)$-cover for the high-variance PMDs.
The cover can be constructed in time $\bigO{n \cdot N\cdot |S|}$ as we try to add one $k$-CRV from $S$ in each step.
We iterate through the cover, and for each element in the cover,
  we need to find the subset $\cals_i \subseteq \cals$ of $(3\eps/5)$-best responses for player $i$,
  and then run the cover construction algorithm again using only the best responses $\{\cals_i\}_{i=1}^n$.
So the overall running time of the algorithm is $\bigO{n N |S|} \cdot \left(\poly(n^k)|S| + \bigO{n N |S|}\right) = n^{k^{O(1/1-c)}}$.
When both $c < 1$ and $k$ are constants, the running time is polynomial in $n$, as claimed in Theorem \ref{thm:main}.
\end{proof}

\subsection{Proof of Proposition \ref{prop:moments-variance-tradeoff}}
This subsection is devoted to the proof of Proposition \ref{prop:moments-variance-tradeoff}.
For two $(n,k)$-PMDs with variance at least $\sigma^2$ in each direction,
  Proposition \ref{prop:moments-variance-tradeoff} gives a quantitative bound on
  how close degree $\ell$ moments need to be (as a function of $\eps$, $\sigma$, $k$ and $\ell$, but independent of $n$),
  in order for the two PMDs to be $\eps$-close in total variation distance.

The proof of Proposition \ref{prop:moments-variance-tradeoff} exploits the sparsity of the continuous Fourier transforms
of our PMDs, as well as careful Taylor approximations of the logarithm of the Fourier transform.
The fact that our PMDs have large variance enables us to take fewer low-degree terms in the Taylor approximation.
For technical reasons, we split our PMD as the sum of $k$ independent component PMDs,
  $X=\sum_{i=1}^k X^i$,
  where all the $k$-CRVs in the component PMD $X^i$ is $i$-maximal.
Because the Fourier transform of $X$ is the product of the Fourier transform of $X^i$,
  we can just bound the pointwise difference between the logarithm of Fourier transform of each component PMD.
One technicality is that since we have no assumption on the variances of the component PMDs $X^i$,
  their Fourier transforms may not be sparse,
  so it is crucial that we bound this difference only on the effective support of the Fourier transform of the entire PMD.

We start by considering a set $S$ that includes the effective support of $X$ (and $Y$ when we show that the means are close):
\begin{lemma}[Essentially Corollary 5.3 of \cite{DiakonikolasKS16stoc}]
\label{lem:clt-support-ellipsoid-and-count}
Let $X$ be an $(n,k)$-PMD with mean $\mu$ and covariance matrix $\Sigma$, such that all the non-zero eigenvalues of $\Sigma$ is at least $\sigma^2$ where $\sigma \geq \poly(1/\eps)$.
Let $S$ be the set of points $x \in \Z^k$ where $(x-\mu)^T \bone=0$ and
\[ (x-\mu)^T (\Sigma+I)^{-1}(x-\mu) \leq (Ck\log(1/\eps)) \;, \]
for some sufficiently large constant $C.$ Then, $X\in S$ with probability at least $1-\eps/2$, and
\[ |S| = \sqrt{\det(\Sigma+I)} \cdot \bigO{\log(1/\eps)}^{k/2}. \]
\end{lemma}
\begin{proof}
Applying Lemma 5.2 of \cite{DiakonikolasKS16stoc},
  we have that $(X-\mu)^T (\Sigma+I)^{-1}(X-\mu) = {O(k\log(k/\eps))}$ with probability at least $1-\eps$.
The set of integer coordinate points in this ellipsoid is the set $S$.
Note that $|S|$ is equal to the volume of
  $S' = \setOfSuchThat{y \in \R^k}{\exists x \in S \text{ with } \norminf{y-x} \le 1/2}$,
  because $S'$ is the disjoint union of cubes of volume $1$, one for each integer point.
But $S'$ is again contained in an ellipsoid with
  $(y-\mu)^T (\Sigma+I)^{-1}(y-\mu) = {O(k\log(k/\eps))}$,
  so $|S| = \mathrm{Vol}(S') = \sqrt{\det(\Sigma + I)} \cdot \bigO{\log(1/\eps)}^{k/2}$.
\end{proof}

Next we show that $\wh{X}$, the Fourier transform of $X$, has a relatively small effective support.
We fold the effective support onto $[0, 1]^k$ to obtain the set $T$.
We use $[x]$ to denote the additive distance of $x \in \R$ to the closest integer,
  i.e., $[x] = \min_{x' \in \Z} \abs{x-x'}$.

\begin{lemma}
\label{lem:clt-fourier-support-and-volume}
Let $X$ be an $(n,k)$-PMD with mean $\mu$ and covariance matrix $\Sigma$,
  such that all the non-zero eigenvalues of $\Sigma$ are at least $\sigma^2$ where $\sigma \ge \poly(k \log(1/\eps))$.
Let $S$ be as above.
Let $\wh{X}$ be the Fourier transform of $X$.
Let $T \eqdef \setOfSuchThat{\xi \in [0,1]^k}{\exists \xi' \in \xi + \Z^k \text{ with } \xi'^T \Sigma \xi' \leq Ck \log(1/\eps)}$,
for some sufficiently large constant $C$. Then, we have that
\begin{itemize}
\item[(i)] For $\xi \in T$, and for all $1 \le i,j \le k$, $[\xi_i - \xi_j] \leq 2\sqrt{Ck \log(1/\eps)}/\sigma$.
\item[(ii)] $\mathrm{Vol}(T)|S| = \bigO{C \log(1/\eps)}^k$.
\item[(iii)] $\int_{[0,1]^k \setminus T} \abs{\wh{X}(\xi)} d\xi \leq \eps/(2|S|)$.
\end{itemize}
\end{lemma}

Lemma~\ref{lem:clt-fourier-support-and-volume} is a technical generalization of Lemma 5.5 of \cite{DiakonikolasKS16stoc}.
Its proof is deferred to Appendix~\ref{app:one}.
This lemma establishes that the contribution to the Fourier transform $\wh{X}$
coming from points outside of $T$ is negligibly small.
We then use the sparsity of the Fourier transform to show that, if two PMDs have Fourier transforms
that are pointwise sufficiently close within the effective support $T$,
then the two PMDs are close in total variation distance.
\begin{lemma}
\label{lem:clt-fourier-close-dtv-close}
Let $X$, $Y$, $S$, $T$ be as above.
If $\abs{\wh{X}(\xi) -\wh{Y}(\xi)} \leq \eps(C'\log(1/\eps))^{-k}$ for all $\xi \in T$
  and a sufficiently large constant $C'$, then $\dtv(X,Y) \leq \eps.$
\end{lemma}
\begin{proof}
For any $x \in \Z^k$, taking the inverse Fourier transform, we have that $\Pr[X=x]=\int_{\xi \in [0,1]^k} e(-\xi \cdot x) \wh{X}(\xi) d\xi$ and similarly $\Pr[Y=x] = \int_{\xi \in [0,1]^k} e(-\xi \cdot x) \wh{Y}(\xi) d\xi$. Thus,
\begin{align*}
\abs{\Pr[X=x] - \Pr[Y=x]} & = \abs{ \int_{\xi \in [0,1]^k} e(-\xi \cdot x) \left(\wh{X}(\xi) - \wh{Y}(\xi)\right) d\xi } \\
& \leq  \int_{\xi \in [0,1]^k} \abs{\wh{X}(\xi) - \wh{Y}(\xi)} d\xi \\
& = \int_{\xi \in T} \abs{\wh{X}(\xi) - \wh{Y}(\xi)} d\xi + \int_{\xi \in [0,1]^k \setminus T} \abs{\wh{X}(\xi) - \wh{Y}(\xi)} d\xi \\
& \leq \mathrm{Vol}(T) \cdot \eps(C'\log(1/\eps))^{-k} + \frac{\eps}{2|S|} \\
& \leq \frac{\bigO{C \log(1/\eps)}^k}{|S|} \cdot \eps(C'\log(1/\eps))^{-k} + \frac{\eps}{2|S|} \\
& \leq \frac{\eps}{|S|}.
\end{align*}
Since $X$ and $Y$ are outside of $S$ each
with probability less than $\eps/2$,
we have that $\dtv(X,Y) \leq \eps/2 + \frac{1}{2}\sum_{x \in S} \abs{\Pr[X=x] - \Pr[Y=x]} \leq \eps$.
\end{proof}

We now have all the ingredients to prove Proposition \ref{prop:moments-variance-tradeoff}.
For two PMDs $X$ and $Y$ that are close in their low-degree moments,
  we show that their Fourier transforms $\wh{X}$ and $\wh{Y}$ are pointwise close on $T$,
  and then by Lemma \ref{lem:clt-fourier-close-dtv-close}, $X$ and $Y$ are close in total variation distance.

\begin{proof}[Proof of Proposition \ref{prop:moments-variance-tradeoff}]
Let $X$, $Y$, $S$, $T$ be as above.
Given Lemma \ref{lem:clt-fourier-close-dtv-close},
  we only need to show that $\forall \xi \in T$, $\abs{\wh{X}(\xi) -\wh{Y}(\xi)} \leq \eps(C'\log(1/\eps))^{-k}$.

Fix $\xi \in T$.
We first examine, without loss of generality, the Fourier transform $\wh{X^k}$ of the $k$-maximal component PMD.
Let $A_k \subseteq [n]$ denote the set of $k$-maximal CRVs.
\begin{align}
  \wh{X^k}(\xi)
 &= \prod_{i \in A_k} \sum_{j=1}^k e(\xi_j) p_{i,j} \nonumber \\
 &= e(|A_k| \xi_k) \prod_{i \in A_k} \left(1 - \sum_{j=1}^{k-1} \left(1-e(\xi_j-\xi_k))p_{i,j}\right) \right) \nonumber \\
 &= e(|A_k| \xi_k) \exp\left(\sum_{i \in A_k} \log\left(1 - \sum_{j=1}^{k-1} \left(1-e(\xi_j-\xi_k))p_{i,j}\right) \right) \right) \nonumber \\
 &= e(|A_k| \xi_k) \exp\left(-\sum_{i \in A_k} \sum_{\ell=1}^{\infty} \frac{1}{\ell} \left(\sum_{j=1}^{k-1} \left(1-e(\xi_j-\xi_k))p_{i,j}\right) \right) \right) \nonumber \\
 &= e(|A_k| \xi_k) \exp\left(-\sum_{m \in \Z^{k-1}_{+}} \binom{\normone{m}}{m} \frac{1}{\normone{m}} M_m (X^k) \prod_{j=1}^{k-1} (1-e(\xi_j-\xi_k))^{m_j} \right) \label{eqn:x-fourier-taylor}
\end{align}
For notational convenience, we use $\Psi^k_X$ to denote the expression inside $\exp(\cdot)$
  in Equation (\ref{eqn:x-fourier-taylor}).
A similar formula holds for the Fourier transform $\wh{X^i}$ and $\wh{Y^i}$ of other $i$-maximal PMDs,
  and we use $\Psi^i_X$ and $\Psi^i_Y$ to denote the corresponding expressions inside $\exp(\cdot)$.
Since the Fourier transform of a PMD is the product of the Fourier transform of its component PMDs, we have
\begin{align*}
  \abs{\wh{X}(\xi) -\wh{Y}(\xi)}
 &= \abs{\prod_{t=1}^k \wh{X^t}(\xi) - \prod_{t=1}^k \wh{Y^t}(\xi)} \\
 &=  \abs{ e\left( \sum_{t=1}^k|A_t| \xi_t \right) \prod_{t=1}^k \left( \exp\left(\Psi^t_X \right) - \exp\left(\Psi^t_Y \right) \right)} \\
 &\le 2 \pi \sum_{t=1}^k \abs{\Psi^t_X - \Psi^t_Y},
\end{align*}
where the last inequality is due to $e(\sum_{t=1}^k |A_t| \xi_t) = 1$,
  and $\abs{\exp(a) - \exp(b)} \le \abs{a-b}$
  if the real parts of $a$ and $b$ satisfy $\text{Re}(a), \text{Re}(b) \le 0$.

So to prove that $\wh{X}(\xi)$ and $\wh{Y}(\xi)$ are pointwise close for all $\xi \in T$,
  it is enough to bound from above $2\pi \sum_{t=1}^k \abs{\Psi^t_X - \Psi^t_Y}$.
We use the fact that $\abs{1-e(\xi_j-\xi_k)} = \bigO{[\xi_j - \xi_k]}$,
  and recall that $[\xi_i - \xi_j] \leq 2\sqrt{Ck \log(1/\eps)}/\sigma$ by Lemma \ref{lem:clt-fourier-support-and-volume}.
We also use the multinomial identity $\sum_{m \in \Z_{+}^{k-1}, \normone{m}=\ell}\binom{\ell}{m} = (k-1)^\ell$.
When $C'$ is a sufficiently large constant, we have
\begin{align*}
\abs{\wh{X}(\xi) -\wh{Y}(\xi)}
&\le 2\pi \sum_{t=1}^k \abs{\Psi^t_X - \Psi^t_Y} \\
&= 2\pi \sum_{t=1}^k \sum_{m \in \Z^{k-1}_{+}} \binom{\normone{m}}{m} \frac{1}{\normone{m}} \abs{M_m (X^t) - M_m (Y^t)} \prod_{j=1}^{k-1} (1-e(\xi_j-\xi_k))^{m_j}  \\
&\le 2\pi \sum_{\ell=1}^{\infty}\frac{(k-1)^\ell}{\ell} \left(\bigO{\frac{\sqrt{k\log(1/\eps)}}{\sigma}}\right)^\ell \sum_{t=1}^k \max_{m \in \Z^{k-1}_{+},\normone{m}=\ell} \abs{M_m(X^t) - M_m(Y^t)} \\
&\le \sum_{\ell=1}^{\infty} k^\ell \left(\frac{C'\sqrt{k\log(1/\eps)}}{2\sigma}\right)^\ell k \cdot \frac{\eps \sigma^\ell}{C'^{k+\ell} \cdot k^{3\ell/2+1} \cdot \log^{k+\ell/2}(1/\eps)} \\
&= \sum_{\ell=1}^\infty 2^{-\ell}  \eps(C'\log(1/\eps))^{-k} \\
&= \eps(C'\log(1/\eps))^{-k}. \qedhere
\end{align*}

\end{proof}

\section{Reductions: Proof of Theorem~\ref{thm:fptas}}
In this section, we show that even a slight improvement of our upper bound
would imply an FPTAS for computing (well-supported) Nash equilibria in anonymous games (Theorem \ref{thm:fptas}).
It is a plausible conjecture that assuming ETH for PPAD, there is no such FPTAS,
in which case our upper bound (Theorem \ref{thm:main}) is essentially tight.

Theorem \ref{thm:fptas} follows directly from the following two lemmas.
Lemma \ref{lem:ane2wsne} converts an $\frac{\eps^2}{4n}$-approximate Nash equilibrium into
  an $\eps$-well-supported Nash equilibrium\footnote{A mixed strategy profile $\vec{s}$ is a \emph{well-supported} Nash equilibrium iff $\forall i \in [n]$, $\forall a, a' \in [k]$, we have $\expectOver{x \sim \vec{s}_{-i}}{u^i_{a}(x)} > \expectOver{x \sim \vec{s}_{-i}}{u^i_{a'}(x)} + \eps \Longrightarrow s_i(a') = 0$, i.e., players can only put non-zero probability on $\eps$-best-response strategies.},
  by reallocating each player's probabilities on strategies with low expected payoffs to the best-response strategy (first observed in \cite{DaskalakisGP09}).
Lemma \ref{lem:adddummy} then uses a padding argument to show that,
  for \emph{$\eps$-well-supported} Nash equilibrium,
  the question of whether there is a polynomial-time algorithm for $\eps = n^{-c}$
  is equivalent for all constants $c > 0$.

\begin{lemma}
\label{lem:ane2wsne}
For any $n$-player game whose payoffs are normalized to be between $[0, 1]$,
if we have an oracle for computing players' payoffs,
we can efficiently convert an $\frac{\eps^2}{4n}$-approximate equilibrium into
an $\eps$-well-supported equilibrium.
\end{lemma}
\begin{proof}
Take an $\frac{\eps^2}{4n}$-approximate equilibrium of the game.
We call a strategy ``good'' for a player if the strategy is an $\frac{\eps}{2}$-best response for the player,
  and we call it ``bad'' otherwise.
A player can put at most probability $\frac{\eps}{2n}$ on the ``bad'' strategies without violating
  the $\frac{\eps^2}{4n}$-approximate equilibrium condition.
We move all the probabilities on ``bad'' strategies for all players to (any one of) their best responses simultaneously.
After moving the probabilities, every player assigns non-zero probabilities only to the ``good'' strategies.
Since the total probability we moved is at most $\frac{\eps}{2}$ and the payoffs are in $[0, 1]$,
the previously ``good'' strategies ($\frac{\eps}{2}$-best responses) are now $\eps$-best responses.
\end{proof}

\begin{lemma}
\label{lem:adddummy}
For $n$-player $k$-strategy anonymous games with $k = O(1)$,
if an $\frac{1}{n^\gamma}$-well-supported equilibrium can be computed in time $O(n^d)$ for constants $\gamma, d > 0$,
then there is an FPTAS for computing approximate-well-supported Nash equilibria in anonymous games.
\end{lemma}
\begin{proof}
Let $\eps$ be the desired quality of the well-supported equilibrium.
If $\frac{1}{n^\gamma} \le \eps$ we are done, so we assume $n$ is smaller.
We set $n' = (1/\eps)^{1/\gamma}$, so that $\frac{1}{n'^\gamma} = \eps$.
Given an $n$-player anonymous game $G$,
we build an $n'$-player anonymous game $G'$ as follows:
we add $n'-n$ dummy players, and give the dummy players utility 1 on strategy 1,
and 0 on any other strategies so in any $\eps$-well-supported equilibria,
the dummy player must all play strategy 1 with probability 1.
(Note that this is only true for $\eps$-well-supported Nash equilibrium; in an $\eps$-approximate Nash equilibrium,
the dummy players can put $\eps$ probability elsewhere.)
We shift the utility function of the actual players to ignore the dummy players on strategy $1$.
Formally, the payoff structure of $G'$ is given by:
\begin{itemize}
\item For each $i > n$,
  \[ u'^i_a(x) = \begin{cases} 1 & \text{if } a = 1 \\ 0 & \text{otherwise} \end{cases} \]
\item For each $i \le n$, we subtract the number of players on strategy 1 by $n' - n$ and then apply the original utility function. We define $\phi: \Z^k \rightarrow \Z^k$ as $\phi(x_1, \ldots, x_k) = (x_1 - (n'-n), x_2, \ldots, x_k)$,
  \[ u'^i_a(x) = \begin{cases} u^i_a(\phi(x)) & \text{if } x_1 \ge n' - n \\ 0 & \text{otherwise} \end{cases} \]
\end{itemize}
Since $\eps = \frac{1}{n'^\gamma}$,
  by assumption we can compute an $\eps$-well-supported equilibrium of $G'$ in time $O(n'^d)$,
  and we can simply remove the dummy players to obtain an $\eps$-equilibrium of the original game $G$.
The running time is $O(n'^d) = \poly(n, 1/\eps)$ when $\gamma = \Theta(1)$.
\end{proof}

\begin{proof}[Proof of Theorem \ref{thm:fptas}]
Assume that we can compute an $\bigO{n^{-c}}$-approximate equilibrium in polynomial time for some constant $c > 1$.
Let $\gamma = c-1$, so we can compute an $\bigO{\frac{1}{n^{1+\gamma}}}$-approximate equilibrium in polynomial time.
By Lemma \ref{lem:ane2wsne},
  we can convert it into an $\bigO{\frac{1}{n^{\gamma/2}}}$-well-supported equilibrium.
Lemma \ref{lem:adddummy} then states that any polynomial-time algorithm that computes a well-supported Nash equilibrium of an inverse polynomial precision gives an FPTAS for computing well-supported Nash equilibria in anonymous games.
\end{proof}

\section{Proof of Theorem~\ref{thm:simple}}
In this section, we present a faster algorithm that computes an
$\tildeO{n^{-1/3} k^{11/3}}$-approximate Nash equilibrium in $n$ player $k$ strategy anonymous games.
Note that this algorithm always runs in polynomial time in the input size,
without assuming any relationship between $n$ and $k$.
 
Our approach builds on the idea of \cite{goldbergT15} to ``smooth'' an anonymous game by forcing all the players to randomize.
We prove that the perturbed game is Lipschitz and therefore admits a pure Nash equilibrium (Lemma \ref{lem:pureeq}),
which corresponds to simple approximate equilibria of a specific form in the original game:
Each player plays one strategy with probability $1-\delta$ for some small $\delta$,
and plays other strategies uniformly at random with probability $\delta$.
To prove the perturbed game is Lipschitz (Proposition \ref{prop:perturb-lipschitz}),
we rely on the recently established multivariate central limit theorem (CLT) of~\cite{DaskalakisDKT15, DiakonikolasKS16stoc}
to show that for $\delta = \Omega(n^{-1/3})$ the associated PMD is close to a discrete Gaussian.

Recall that an anonymous game $G = (n, k, \{u^i_a\}_{i \in [n], a \in [k]})$ is $\lambda$-Lipschitz if
\[ \forall i \in [n], \forall a \in [k], \quad \forall x, y \in \Pi^k_{n-1}, \quad \abs{u^i_a(x) - u^i_a(y)} \le \lambda \normone{x-y}. \]
An approximate pure Nash equilibrium always exists in Lipschitz anonymous games.
  
\begin{lemma}[\cite{DaskalakisP15,AzrieliS13}]
\label{lem:pureeq}
Every $\lambda$-Lipschitz anonymous game with $k$ strategies admits a $(2 k \lambda)$-approximate pure Nash equilibrium.
Moreover, such an approximate equilibrium can be found in time $\tildeO{n+k}$ times the description size of the game.
\end{lemma}

We perturb the input game $G$ to get another game $G_\delta$ as follows.
Let $X_\delta(e_j)$ denote the $k$-CRV that takes value $e_j$ with probability $1-\delta$,
  and takes value $e_{j'}$ with probability $\frac{\delta}{k-1}$ for all other $j' \neq j$.
When a player plays the strategy $j$ in the perturbed game $G_\delta$,
  it is as if she is playing $X_\delta(e_j)$ in the original game $G$.
For example, the strategy $(1, 0, \ldots, 0)$ in $G_\delta$
  maps back to the mixed strategy $(1-\delta, \frac{\delta}{k-1}, \ldots, \frac{\delta}{k-1})$ in $G$.

By forcing all players to randomize,
  we increase the uncertainty in the outcome of the game (i.e., the variance of the resulting PMD),
  and thus making the game ``smoother''.
As we will prove later, the perturbed game $G_\delta$ is $\lambda$-Lipschitz for
  $\lambda = \tildeO{\frac{k^{9/2}}{\sqrt{n\delta}}}$.
It then follows from Lemma \ref{lem:pureeq} that there exists
  a $(2k\lambda)$-pure Nash equilibrium of $G_\delta$,
  which is a $(\delta + 2k\lambda)$-mixed Nash equilibrium of $G$.
The next proposition formally defines the payoff structure of $G_\delta$,
  and bounds its Lipschitz constant.

\begin{proposition}
\label{prop:perturb-lipschitz}
Given an anonymous game $G = (n, k, \{u^i_a\}_{i\in [n],a\in [k]})$ with payoffs normalized to $[0,1]$,
we define an anonymous game $G_\delta = (n, k, \{u'^i_a\}_{i\in [n],a\in [k]})$ as follows,
\[ \forall i\in [n], a\in [k], x \in \Pi^k_{n-1}, \quad u'^i_a(x) := (1-\delta)\expectOver{x' \sim M_\delta(x)}{u^i_{a}(x')} + \frac{\delta}{k-1} \sum_{a' \neq a}\expectOver{x' \sim M_\delta(x)}{u^i_{a'}(x')}, \]
where $M_\delta(x) = \sum_{j\in[k]} x_j X_\delta(e_j)$ is an $(n-1,k)$-PMD
that corresponds to the perturbed outcome of the partition $x \in \Pi^k_{n-1}$.
Then $G_\delta$ is $\tildeO{\frac{k^{9/2}}{\sqrt{n\delta}}}$-Lipschitz.
\end{proposition}

We defer the proof of Proposition \ref{prop:perturb-lipschitz} to the next subsection.
We now show how Theorem \ref{thm:simple} follows from Proposition \ref{prop:perturb-lipschitz}.

\begin{proof}[Proof of Theorem \ref{thm:simple}]
Proposition \ref{prop:perturb-lipschitz} shows that $G_\delta$ is $\tildeO{\frac{k^{9/2}}{\sqrt{n\delta}}}$-Lipschitz.
By Lemma \ref{lem:pureeq}, there exists a $(2k\lambda)$-approximate pure Nash equilibrium in $G_\delta$,
and as noted in \cite{DaskalakisP15}, such an approximate equilibrium can be found
in total number of bit operations that is $\tildeO{n+k}$ times the description size of $G_\delta$,
by enumerating pure strategy profiles and solving maximum flows to match players to mixed strategies.
Since we can compute the payoff structure of $G_\delta$ in polynomial-time given the input game $G$,
the overall running time is polynomial in the input size.

We now bound the quality of the approximate Nash equilibrium.
Note that a $(2k\lambda)$-pure equilibrium of $G_\delta$ is a $(\delta + 2k\lambda)$-mixed Nash equilibrium of $G$,
  since an $\eps$-equilibrium in $G_\delta$ means that players cannot gain more than $\eps$
  by deviating to the mixed strategies of the form $X_\delta(e_j) = (1-\delta)e_j + \frac{\delta}{k-1} (\bone - e_j)$,
  so they gain at most ($\delta + 2k\lambda$) by deviating to any $e_j$.
Because changing what a player is doing $\delta$ fraction of the time
  can change her payoff by at most $\delta$.
Therefore, we can compute an $(\delta + 2k\lambda) = \tildeO{\delta + \frac{k^{11/2}}{\sqrt{n\delta}}}$-equilibrium
  of the original game $G$ in polynomial-time for any $\delta > 0$.
Finally, setting $\delta = \frac{k^{11/3}}{n^{1/3}}$, we get an $\tildeO{\frac{k^{11/3}}{n^{1/3}}}$-approximate Nash equilibrium.
\end{proof}

\subsection{Proof of Proposition~\ref{prop:perturb-lipschitz}}
\label{sec:perturb-lipschitz}
This section is devoted to the proof of Proposition \ref{prop:perturb-lipschitz}.
We will make use of the following two results.
The first lemma is the multivariate central limit theorem from \cite{DiakonikolasKS16stoc},
which states that if an $(n,k)$-PMD $X$ has high variance in all directions orthogonal to the all ones vector $\bone$
(its variance along $\bone$ is 0),
then the projection of $X$ on the first $(k-1)$ coordinates
is close to a discretized Gaussian distribution with the same mean vector and covariance matrix.

\begin{lemma}[\cite{DiakonikolasKS16stoc}]
\label{lem:clt}
Let $X$ be an $(n, k)$-PMD, and $X'$ be a $(k-1)$-dimensional random variable
that is the projection of $X$ onto its first $k-1$ coordinates.
Let $\Sigma'$ be the covariance matrix of $X'$.
Suppose that $\Sigma'$ has no eigenvectors with eigenvalue less than $\sigma'^2$.
Let $G'$ be the distribution obtained by sampling from $\caln(\expect{X'}, \Sigma')$
and rounding to the nearest point in $\Z^k.$
Then, we have that
\[ \totalvardist{X'}{G'} \le \bigO{k^{7/2} \sqrt{\log^3(\sigma')}/\sigma'}. \]
\end{lemma}

The second simple lemma states that if two $k$-dimensional Gaussian
distributions have similar mean vectors and variances (in all directions),
then they are close in total variation distance.

\begin{lemma}[\cite{DaskalakisDKT15}]
\label{lem:dtvgaussian}
For two $k$-dimensional Gaussians $X \sim \caln(\mu_1, \Sigma_1)$ and $Y \sim \caln(\mu_2, \Sigma_2)$,
  such that for all unit vector $v$,
\[ \abs{v^T (\mu_1 - \mu_2)} \le \eps s_v, \text{ and} \quad \abs{v^T (\Sigma_1 - \Sigma_2) v} \le \frac{\eps s_v^2}{2 \sqrt{k}}, \]
where $s^2_v = \max\{v^T \Sigma_1 v, v^T \Sigma_2 v\}$.
Then $\totalvardist{X}{Y} \le \eps$.
\end{lemma}

\begin{proof}[Proof of Proposition \ref{prop:perturb-lipschitz}]
To prove the game $G_\delta$ is $\lambda$-Lipschitz, we need to show that
\[ \forall i \in [n], \forall a \in [k], \quad \forall x, y \in \Pi^k_{n-1},
    \quad \abs{\expectOver{x'\sim M_\delta(x)}{u^i_a(x')} - \expectOver{y'\sim M_\delta(y)}{u^i_a(y')}} \le \lambda \normone{x-y}. \]
In fact, because the payoff entries are normalized in $[0,1]$, it is sufficient to show that
the total variation distance between the $(n-1,k)$-PMDs $M_\delta(x)$ and $M_\delta(y)$ is small, namely
\[ \quad \forall x, y \in \Pi^k_{n-1}, \quad \totalvardist{M_\delta(x)}{M_\delta(y)} \le \lambda \normone{x-y}. \]
Let $M'_\delta(x)$ and $M'_\delta(y)$ be the distributions $M_\delta(x)$ and $M_\delta(y)$ projected onto their first $k-1$ coordinates.
Note that since all coordinates must sum to $n$, the $k^{th}$ coordinate is redundant and so $\totalvardist{M_\delta(x)}{M_\delta(y)}=\totalvardist{M'_\delta(x)}{M'_\delta(y)}$.
To show that $M'_\delta(x)$ and $M'_\delta(y)$ are close in total variation distance,
we first prove that the covariance matrix of $M'_\delta(x)$ has high variance in all directions,
which allows us to use the multivariate central limit theorem (Lemma \ref{lem:clt}) to conclude that
both $M'_\delta(x)$ and $M'_\delta(y)$ are close to the (discretized) Gaussian distributions
with the same mean vectors and covariance matrices respectively.
We then bound from above the total variation distance between
two high-variance $k$-dimensional Gaussian distributions
whose mean vectors are essentially $x$ and $y$.

Recall that $M_\delta(x)$ is the sum of $n-1$ independent $k$-CRVs,
and let $\Sigma_1$ denote the covariance matrix of $M_\delta(x)$.
For any unit vector $v \in \R^k$ that is orthogonal to the all-one vector, we have
\begin{align*}
  \var[v^T X_\delta(e_j)] &= \expect{\left( v^T X_\delta(e_j) \right)^2} - \left(\expect{v^T X_\delta(e_j)}\right)^2 \\
    &= (1-\delta) v^2_j + \frac{\delta}{k-1} \sum_{j' \neq j} v^2_{j'}
      - \left((1-\delta) v_j + \frac{\delta}{k-1} \sum_{j' \neq j} v_{j'} \right)^2 \\
    &= (1-\delta) v^2_j + \frac{\delta}{k-1} (1-v^2_j) - \left((1-\delta) v_j - \frac{\delta}{k-1} v_j \right)^2 \\
    &\ge \frac{\delta}{k-1},
\end{align*}
where we simplify the expression using the fact that $\sum_j v_j = 0$ and $\sum_j v^2_j = 1$,
  and then take derivative to minimize it.
Therefore, for any unit vector $v$, $v^T \Sigma_1 v = \var[v^T M_\delta(x)] = \sum_{j\in [k]} x_j \var[v^T X_\delta(e_j)] \ge \frac{(n-1)\delta}{k-1}$, which implies that $\Sigma_1$ has no eigenvalues less than $\frac{(n-1)\delta}{k-1}$ (except the one associated with $\bone$).
We then use the following lemma to bound from below the eigenvalues of $\Sigma'_1$:
\begin{lemma} Suppose that $\Sigma$ is a positive semidefinite matrix with $\Sigma \bone = 0$ and that all other eigenvalues of $\Sigma$ are at least $\sigma^2$.
Then for all vectors $w \in \R^k$ with $w_k = 0$, we have that
\[ \frac{w^T \Sigma w}{w^Tw} \ge \sigma^2/k. \]
\end{lemma}
\begin{proof}
Let $w$ be a vector that minimizes $\frac{w^T \Sigma w}{w^T w}$ over $w \in \R^k$ with $w_k = 0$.
Then $v = w - \frac{w^T \bone}{k} \bone$ has $v^T \bone = 0$ and so $v^T \Sigma v \geq \sigma^2 v^T v$.
We have $v^T \Sigma v=w^T \Sigma w$ since $v - w$ is a multiple of $\bone$, and we have
\begin{align*}
v^T v & = \left(w - \frac{w^T \bone}{k} \bone\right)^T \left(w - \frac{w^T \bone}{k} \bone\right) \\
& = w^T w + (w^T \bone)^2/k - 2(w^T \bone)^2/k \\
& = \|w\|_2^2 - \|w\|_1^2/k \\
& \geq w^T w / k \;,
\end{align*}
where the last line follows from the inequality $\|w\|_1 \leq \sqrt{k-1}\|w\|_2$.
Thus, we have that
\[ \frac{w^T \Sigma w}{w^Tw} \ge \frac{v^T \Sigma v}{k v^Tv} \ge \sigma^2 / k. \qedhere \]
\end{proof}
Since all except one eigenvalues of each of $\Sigma_1$ and $\Sigma_2$
are at least $\frac{(n-1)\delta}{(k-1)}$,
the minimum eigenvalues of $\Sigma'_1$ and $\Sigma'_2$
are at least $\frac{(n-1)\delta}{k^2}$. Let $\mathcal{Z}(\mu,\Sigma)$
be the discretized Gaussian obtained by rounding $\caln(\mu,\Sigma)$
to the nearest integer in every coordinate.
Then, by Lemma \ref{lem:clt}, we have
\begin{align}
\label{eqn:clt}
\totalvardist{M'_\delta(x)}{\mathcal{Z}(\mu'_1,\Sigma'_1)} \le \tildeO{\frac{k^{9/2}}{\sqrt{n \delta}}}, \quad
\totalvardist{M'_\delta(y)}{\mathcal{Z}(\mu'_2,\Sigma'_2)} \le \tildeO{\frac{k^{9/2}}{\sqrt{n \delta}}}.
\end{align}
Next, we use Lemma \ref{lem:dtvgaussian} to bound the total variation distance
between the $k$-dimensional Gaussian distributions $\caln(\mu'_1,\Sigma'_1)$ and $\caln(\mu'_2,\Sigma'_2)$.
Let $\mu_1, \mu_2$ and $\Sigma_1, \Sigma_2$ be the mean vectors
and the covariance matrices of $M_\delta(x)$ and $M_\delta(y)$ respectively.
Observe that
\[ \mu_1 = \left(1-\frac{k\delta}{k-1}\right)x + \delta \bone. \]
So, for any unit vector $v \in \R^k$,
\begin{align*}
  s^2_v &= \max\{v^T \Sigma_1 v, v^T \Sigma_2 v\} \ge \frac{(n-1)\delta}{k-1}, \\
  \abs{v^T(\mu_1-\mu_2)}
    &= \abs{v^T \left(\left(1-\frac{k\delta}{k-1}\right)x - \left(1-\frac{k\delta}{k-1}\right)y\right)}
    \le \left(1-\frac{k\delta}{k-1}\right) \normone{x-y} \le \normone{x-y}.
\end{align*}
If the unit vector $v$ is orthogonal to $\bone$, we can use the expression for $v^T \Sigma_1 v$ we had earlier.
Taking derivative with respect to $v_j$ shows that $\var[v^T X_\delta(e_j)]$ is maximized at $v = \pm e_j$.
Hence, we can write
\begin{align*}
  \abs{v^T(\Sigma_1-\Sigma_2)v}
   &= \sum_{j\in [k]} (x_j - y_j) \var[v^T X_\delta(e_j)]
    \le \normone{x-y} \max_j \var[v^T X_\delta(e_j)] \\
   &\le \normone{x-y}  \left[(1-\delta) - \left(1-\frac{k\delta}{k-1}\right)^2\right]
    = \frac{k-1}{k+1}\delta \normone{x-y} \le \delta \normone{x-y}.
\end{align*}
To see that the upper bound on $\abs{v^T(\Sigma_1-\Sigma_2)v}$ holds for all unit vectors,
observe that for both covariance matrices it holds $\Sigma_1 \bone = \Sigma_2 \bone = \vec{0}$.
For any unit vector $v'$, we can take its projection onto the subspace orthogonal to $\bone$,
and write $v'$ as a linear combination $\alpha v + \beta \bone$,
for some $\alpha < 1$ and a unit vector $v$ that is orthogonal to $\bone$. That is,
\[ \abs{v'^T (\Sigma_1 - \Sigma_2) v'} = \abs{(\alpha v + \beta \bone)^T (\Sigma_1 - \Sigma_2) (\alpha v + \beta \bone)}
     = \abs{\alpha^2 v^T (\Sigma_1 - \Sigma_2) v} \le \abs{v^T (\Sigma_1 - \Sigma_2) v}. \]
Thus, for all unit vectors $v \in \R^k$, we have
$\abs{v^T(\mu_1-\mu_2)} \leq \normone{x-y}$ and $\abs{v^T (\Sigma_1 - \Sigma_2) v} \leq \delta \normone{x-y}$.
In particular, this holds for vectors with $k^{th}$ coordinate $0$.
Hence, for all $v \in \R^{k-1}$, we have $\abs{v^T(\mu'_1-\mu'_2)} \leq \normone{x-y}$ and $\abs{v^T (\Sigma'_1 - \Sigma'_2) v} \leq \delta \normone{x-y}$.
	
Finally, we set $\eps = \bigO{\frac{\sqrt{k}}{\sqrt{n\delta}} + \frac{k^{3/2}}{n}} \normone{x-y}$
  to satisfy the requirements of Lemma \ref{lem:dtvgaussian}, and therefore $\totalvardist{\caln(\mu'_1,\Sigma'_1)}{\caln(\mu'_2,\Sigma'_2)} \leq \eps$.
By the data processing inequality,
rounding both distributions to the nearest integer coordinates
does not increase their total variation distance, therefore
\begin{align}
\label{eqn:dtvgaussian}
  \totalvardist{\mathcal{Z}(\mu'_1,\Sigma'_1)}{\mathcal{Z}(\mu'_2,\Sigma'_2)} \le \eps.
\end{align}
By the triangle inequality, Equations (\ref{eqn:clt}) and (\ref{eqn:dtvgaussian}) yield
\begin{align*}
\totalvardist{M_\delta(x)}{M_\delta(y)} & = \totalvardist{M'_\delta(x)}{M'_\delta(y)} \\
  &\le \totalvardist{M'_\delta(x)}{\mathcal{Z}(\mu'_1,\Sigma'_1)} + \totalvardist{\mathcal{Z}(\mu'_1,\Sigma'_1)}{\mathcal{Z}(\mu'_2,\Sigma'_2)} + \totalvardist{\mathcal{Z}(\mu'_2,\Sigma'_2)}{M'_\delta(y)} \\
  &\le \tildeO{\frac{k^{9/2}}{\sqrt{n\delta}}} + \bigO{\frac{\sqrt{k}}{\sqrt{n\delta}} + \frac{k^{3/2}}{n}} \normone{x-y} \\
  &\le \tildeO{\frac{k^{9/2}}{\sqrt{n\delta}}} \normone{x-y}.
\end{align*}
The last inequality holds because $M_\delta(x) = M_\delta(y)$ when $x = y$,
so we can assume that $\normone{x-y} \ge 1$.
This concludes the proof that $G_\delta$ is $\lambda$-Lipschitz
for $\lambda = \tildeO{\frac{k^{9/2}}{\sqrt{n\delta}}}$.
\end{proof}

\bibliographystyle{alpha}
\bibliography{refs}

\begin{appendix}

\section*{Appendix}

\section{Proof of Lemma~\ref{lem:clt-fourier-support-and-volume}} \label{app:one}
This lemma is a generalization of Lemma 5.5 of \cite{DiakonikolasKS16stoc},
which assumes that $\eps=\tilde{O_k}(1/\sigma)$.
Thus, we need to be careful about where this relation was used in the proof.

Note that for a fixed $\xi$, if $\xi'$ satisfies $\xi' \in \xi + \Z^k$ and
$\xi'^T \Sigma \xi' \leq Ck \log(1/\eps)$, then so does $\xi+i\bone$ for all $i \in \Z$.
We define $T'$ as
$T' \eqdef \setOfSuchThat{\xi' \in \R^k}{\xi'^T \Sigma \xi' \leq Ck \log(1/\eps) \text{ and } 0 \leq \xi' \cdot \bone \leq k}$.
Then, $\xi \in T$ if and only if there is a $\xi' \in T'$ with $\xi-\xi' \in \Z^k$.

\begin{enumerate}
\item[(i)]
Because $\xi - \xi' \in \Z^k$, we have $[\xi_i - \xi_j] \le \abs{\xi'_i - \xi'_j}$.
So to prove (i), we need to show that
  $\abs{\xi'_i-\xi'_j} \leq 2\sqrt{Ck \log(1/\eps)}/\sigma$ for all $\xi' \in T'$, $i$ and $j$.

Fix $\xi' \in T'$, we define $\tilde \xi'$ to be the projection of $\xi'$ onto the subspace orthogonal to $\bone$,
  i.e., $\tilde \xi' = \xi'- \frac{\xi' \cdot \bone}{k} \bone$.
Since $\Sigma \bone = \vec{0}$ and all other eigenvalues of $\Sigma$ are at least $\sigma^2$,
  for all $i$, $j$ we have
\[ \abs{\xi'_i - \xi'_j} = \abs{\tilde \xi'_i - \tilde \xi'_j} \le 2 \norminf{\tilde \xi'} \le 2 \normtwo{\tilde \xi'}
  \le 2\sqrt{\xi'^T \Sigma \xi' / \sigma^2} \le 2\sqrt{Ck\log(1/\eps)}/\sigma. \]
This proves (i).

\item[(ii)]
Next we consider $\mathrm{Vol}(T')$.
If $\xi' \in T'$, we know that $\|\xi'-(\xi' \cdot \bone/k)\bone\|_2^2 \leq Ck \log(1/\eps)/\sigma^2$.
Also $0 \leq \xi' \cdot \bone \leq k$ implies that $\|(\xi' \cdot \bone/k)\bone\|_2^2 \leq k$.
Because these two vectors are orthogonal, we can write
\[ \|\xi'\|_2^2 = \left\|  \xi'-\frac{\xi' \cdot \bone}{k}\bone \right\|_2^2 + \left\|\frac{\xi' \cdot \bone}{k}\bone \right\|_2^2 \leq Ck \log(1/\eps)/\sigma^2 + k \le 2Ck \log(1/\eps), \]
where the last inequality holds by the assumption that $\sigma \ge 1$.
Thus, \[\xi'^T (\Sigma+I) \xi' = \xi'^T \Sigma \xi' +  \|\xi'\|_2^2 \leq 3Ck\log(1/\eps).\]
By Claim 5.4 of~\cite{DiakonikolasKS16stoc},
we get that \[\mathrm{Vol}(T') \leq \det(\Sigma+I)^{-1/2} \bigO{C \log(1/\eps)}^{k/2}.\]
It then follows from Lemma \ref{lem:clt-support-ellipsoid-and-count} that
$\mathrm{Vol}(T')|S| = \bigO{C \log(1/\eps)}^k$.

To show (ii), we need to show that $\mathrm{Vol}(T) \leq \mathrm{Vol}(T')$.
Note that $T'$ is a disjoint union of its intersections with unit cubes with integer corners, and so
\[ \mathrm{Vol}(T') = \sum_{b \in \Z^k} \mathrm{Vol}\left(T' \cap \prod_{i=1}^k [b_i,b_i+1)\right).\]
On the other hand, $T$ is the union of translations of these sets
\[ T = \bigcup_{b \in \Z^k} \{\xi'-b : \xi \in T' \cap \prod_{i=1}^k [b_i,b_i+1)\} \;, \]
so $\mathrm{Vol}(T) \leq \mathrm{Vol}(T')$.

\item[(iii)]
By the pigeonhole principle,
for every $\xi \in \R^k$,
there is an interval $I_\xi$ of length $\frac{k}{k+1}$
such that there exists $\xi' \in \xi+\Z^k$ where all the coordinates of $\xi'$ are in $I_\xi.$
We define $T_m$ to be
\[ T_m \eqdef \setOfSuchThat{\xi \in [0,1]^k}{\exists \xi' \in \left(\xi + \Z^k\right) \cap I_\xi^k \quad \text{and} \quad 2^{m} C k \log(\sigma) \leq \xi'^T \Sigma \xi' \leq 2^{m+1} C k \log(\sigma)} \;. \]
Then, we have that $T \cup \left( \bigcup_{m=0}^\infty T_m \right) = [0,1]^k$,
although these sets need not be disjoint.
Thus, $[0,1]^k/T \subseteq \bigcup_{m=0}^\infty T_m$ and so
\[ \int_{[0,1]^k/T} \abs{\wh{X}(\xi)} d\xi \leq \sum_{m=0}^\infty \mathrm{Vol}(T_m)\sup_{\xi\in T_m}|\wh{X}(\xi)|. \]
If we apply (ii) of this lemma with $2^{m+1} C$ instead of $C$,
the resulting set $T$ would be a superset of $T_m$.
Thus, we have that $\mathrm{Vol}(T_m) \leq \bigO{2^{m+1} C \log(1/\eps)}^k/|S|$.
To show (iii), we bound $\sup_{\xi\in T_m}|\wh{X}(\xi)|$ using the following claim,
  which gives a ``Gaussian decay'' upper bound on the magnitude of the Fourier transform.

\begin{claim}
For $\xi \in T_m,$ it holds $|\wh{X}(\xi)| \leq \exp(-\Omega(C 2^m \log(1/\eps)/k)).$
If additionally we have $m \leq 3 \log_2 k,$ then $|\wh{X}(\xi)| = \exp(-\Omega(C 2^m k \log(1/\eps)))$.
\end{claim}
\begin{proof}
We take $\xi' \in \left(\xi + \Z^k\right) \cap I_\xi^k$ as in the definition of $T_m$.
Lemma 3.10 of \cite{DiakonikolasKS16stoc} gives that if the coordinates of $\xi'$ lie in an interval of length $1-\delta$,
then
\[ |\wh{X}(\xi)|= |\wh{X}(\xi')| \leq  \exp(-\Omega(\delta^2 \xi'^T \cdot \Sigma \cdot \xi' )) = \exp(-\Omega(C 2^m k \log(1/\eps) \delta^2)). \]
By the definition of $T_m$, we take $\delta = \frac{1}{k+1}$
to get the bound $|\wh{X}(\xi)| \leq \exp(-\Omega(C 2^m \log(1/\eps)/k))$.

To get the stronger bound, we need to show that when $m$ is small
all coordinates of $\xi'$ are in a shorter interval.
This is because, if we apply (i) of this lemma with $2^{m+1} C$ instead of $C$, we have
$|\xi'_i - \xi'_j| \leq \sqrt{ 2^{m+3} Ck \log(1/\eps)}/\sigma$ for any $i$, $j$.
When $m \leq \log_2( \sigma/(Ck \log(1/\eps))) - 4$,
we can take $\delta = 1/2$ and obtain the stronger bound of the claim.

This is where we use our assumption that $\sigma \geq \poly(k \log(1/\eps))$.
We need $m \le 3 \log_2 k \leq \log_2( \sigma/(Ck \log(1/\eps))) - 4$,
which holds when $\sigma \geq 16 Ck^4 \log (1/\eps)$.
\end{proof}
Finally, for (iii) we can write
\begin{align*}
\int_{[0,1]^k/T} \abs{\wh{X}(\xi)} d \xi
& \leq \sum_{m=0}^\infty \mathrm{Vol}(T_m)\sup_{\xi\in T_m}|\wh{X}(\xi)| \\
& \leq \sum_{m=0}^\infty \bigO{2^{m+1} C \log(1/\eps)}^k \sup_{\xi\in T_m}|\wh{X}(\xi)| \\
& \leq \frac{\bigO{ C \log(1/\eps)}^k}{|S|} \sum_{m=0}^{\infty} 2^{mk} \sup_{\xi\in T_m}|\wh{X}(\xi)| \;.
\end{align*}
We divide this sum into two pieces:
\begin{eqnarray*}
\sum_{m=0}^{3 \log_2 k} 2^{mk} \sup_{\xi\in T_m}|\wh{X}(\xi)| & \leq & \sum_{m=0}^{3 \log_2 k} 2^{mk} \exp(-\Omega(C 2^m k \log(1/\eps))) \\
& \leq & \sum_{m=0}^{3 \log_2 k} \exp(-\Omega(C (2^m-m) k \log(1/\eps))) \\
& \leq & \sum_{m=0}^{3 \log_2 k} 2^{-m} \exp(-\Omega(C k \log(1/\eps))) \\
& \leq & \exp(-\Omega(C k \log(1/\eps))) = \eps^{\Omega(Ck)} \;,
\end{eqnarray*}
and
\begin{align*}
\sum_{m=3 \log_2 k}^{\infty} 2^{mk} \sup_{\xi\in T_m}|\wh{X}(\xi)|  \leq & \sum_{m=3 \log_2 k}^{\infty}  2^{mk} \exp(-\Omega(C 2^m  \log(1/\eps)/k)) \\
 \leq & \sum_{m=3 \log_2 k}^{\infty}  \exp(-\Omega(C (2^m-mk^2) \log(1/\eps)/k)) \\
 \leq & \sum_{m=3 \log_2 k}^{\infty}  \exp(-\Omega(C (k^2+mk) \log(1/\eps)/k)) \\
 \leq & \sum_{m=3 \log_2 k}^{\infty}  2^{-m} \exp(-\Omega(C k \log(1/\eps)))
 \leq  \eps^{\Omega(Ck)} \;.
\end{align*}
We thus have $\int_{[0,1]^k \setminus T} |\wh{X}(\xi)|d\xi \leq O(C \log(1/\eps))^{k} \eps^{\Omega(Ck)}/|S| \le \eps/(2|S|)$. \qedhere
\end{enumerate}

\end{appendix}

\end{document}